\documentclass[10pt,a4paper]{article}

\usepackage{amsmath,amsgen,latexsym}
\usepackage{amstext,amssymb,amsfonts,latexsym}
\usepackage{theorem}
\usepackage{pifont}
\usepackage{microtype}
\usepackage{graphicx}

 \newcommand{\bs}{\bigskip}
 \newcommand{\ms}{\medskip}
 \newcommand{\n}{\noindent}
 \newcommand{\s}{\smallskip}
 \newcommand{\hs}[1]{\hspace*{ #1 mm}}
 \newcommand{\vs}[1]{\vspace*{ #1 mm}}

 \newcommand{\setempty}{\varnothing}

 \newcommand{\nat}{\mathbb{N}}
 
 \newcommand{\integer}{\mathbb{Z}}

 \newcommand{\co}{\mathrm{co}\mbox{-}}

 \newcommand{\CC}{{\cal C}}
 \newcommand{\FF}{{\cal F}}

 \newcommand{\MM}{{\cal M}}
 \newcommand{\SSS}{{\cal S}}
 \newcommand{\TT}{{\cal T}}
 \newcommand{\PP}{{\cal P}}

 \newcommand{\reg}{\mathrm{REG}}
 \newcommand{\cfl}{\mathrm{CFL}}
 
 \newcommand{\dcfl}{\mathrm{DCFL}}

\theoremstyle{plain}
\theoremheaderfont{\bfseries}
 \newtheorem{theorem}{Theorem}[section]
 \newtheorem{lemma}[theorem]{Lemma}
 \newtheorem{proposition}[theorem]{{\bf Proposition}}
 \newtheorem{corollary}[theorem]{Corollary}

 {\theorembodyfont{\rmfamily}
  }
 {\theorembodyfont{\rmfamily} }
 {\theorembodyfont{\rmfamily} }

 \newtheorem{claim}{Claim}

 \newenvironment{proof}{\par \noindent
            {\bf Proof. \hs{2}}}{\hfill$\Box$ \vspace*{3mm}}

 \newenvironment{proofof}[1]{\vspace*{5mm} \par \noindent
         {\bf Proof of #1.\hs{2}}}{\hfill$\Box$ \vspace*{3mm}}

 \newcommand{\ceilings}[1]{\lceil #1 \rceil}
 \newcommand{\floors}[1]{\lfloor #1 \rfloor}
 \newcommand{\pair}[1]{\langle #1 \rangle}

\newcommand{\ignore}[1]{}

\newcommand{\cent}{|\!\! \mathrm{c}}
\newcommand{\dollar}{\$}

\newcommand{\klda}[1]{#1 \mbox{-}\mathrm{LDA}}

\begin{document}

\pagestyle{plain}
\setcounter{page}{1}


\begin{center}
{\Large {\bf Intersection and Union Hierarchies of Deterministic \s\\
Context-Free Languages and Pumping Lemmas}}\footnote{This exposition completes and corrects a preliminary report that appeared in the Proceedings of the 14th International Conference on Language and Automata Theory and Applications (LATA 2020), Lecture Notes in Computer Science, vol. 12038, pp. 341--353, 2020. A conference talk was given online during October 20--24, 2021 due to the coronavirus pandemic.} \bs\\
{\sc Tomoyuki Yamakami}\footnote{Present Affiliation: Faculty of Engineering, University of Fukui, 3-9-1 Bunkyo, Fukui 910-8507, Japan} \bs\\
\end{center}


\begin{abstract}
We study the computational complexity of finite intersections and finite unions of deterministic context-free (dcf) languages. Earlier, Wotschke [J. Comput. System Sci. 16 (1978) 456--461] demonstrated that intersections of $(d+1)$ dcf languages are in general more powerful than intersections of $d$ dcf languages for any positive integer $d$ based on the separation result of the intersection hierarchy of Liu and Weiner [Math. Systems Theory 7 (1973) 185--192]. The argument of Liu and Weiner, however, works only on bounded languages of particular forms, and therefore Wotschke's result is not directly extendable to other non-bounded languages.
To deal with a wide range of languages for the non-membership to the intersection hierarchy, we circumvent the specialization of their proof technics and devise a new and practical technical tool: two  pumping lemmas for finite unions of  dcf languages.
Since the family of dcf languages is closed under complementation and also under intersection with regular languages, these pumping lemmas help us establish the non-membership relation of languages formed by finite intersections of target languages. We also concern ourselves with a relationship to deterministic limited automata of Hibbard [Inf. Control 11 (1967) 196--238] in this regard.

\ms

\n{\bf Key words.}
deterministic pushdown automata, intersection and union hierarchies, pumping lemma, limited automata, nondegenerate iterative pair, state-stack pair, LR(1) grammar
\end{abstract}

\sloppy
\section{A Historical Account and New Separation Results}\label{sec:introduction}

We briefly review a historical account of the subject of this exposition and provide a quick overview of the main contribution.

\subsection{Historical Background and Motivational Discussion}\label{sec:intersection-hierarchy}

In formal languages and automata theory, \emph{context-free languages} constitute a fundamental language family, $\cfl$, which is situated in between the family $\reg$ of regular languages and that of context-sensitive languages in Chomsky's hierarchy.
Over a half century, a flurry of studies have been conducted toward a better understanding of the nature of the family $\cfl$. It has been well known that, among numerous structural properties, $\cfl$ enjoys a closure property under the union operation but not the intersection operation since the non-context-free language $L_{abc} =\{a^{n}b^{n}c^{n} \mid n\geq0\}$, for instance, can be expressed as the intersection of the two simple context-free languages $A_1=\{a^nb^nc^p\mid n,p\geq0\}$ and $A_2=\{a^nb^pc^p\mid n,p\geq0\}$. The latter non-closure property can be further generalized to any intersection of $d$ ($\geq1$) context-free languages.
For later notational convenience, we here write $\cfl(d)$ for the family of such intersection languages, in other words, the $d$ intersection closure of $\cfl$ (refer to, e.g., \cite{Yam14a} for more information on this notation).
With the use of this succinct notation, the above language $L_{abc}$ belongs to the difference $\cfl(2) - \cfl$. In a similar way, the more complicated language $L_d=\{a_1^{n_1} a_2^{n_2} \cdots a_d^{n_d} b_1^{n_1} b_2^{n_2} \cdots b_d^{n_d} \mid n_1,n_2,\ldots,n_d\geq0\}$ over the alphabet $\{a_1,a_2,\ldots,a_d,b_1,b_2,\ldots,b_d\}$ falls into $\cfl(d)$ because $L_d$ can be expressed as an intersection of $d$ context-free languages of the form $\{a_1^{n_1} a_2^{n_2} \cdots a_d^{n_d} b_1^{m_1} b_2^{m_2}  \cdots a_d^{m_d}\mid n_k=m_k \}$ for each fixed index $k$ between $1$ and $d$, where $n_1,n_2,\ldots,n_d,m_1,m_2,\ldots,m_d\geq0$.
In 1973, Liu and Weiner \cite{LW73} (who actually used a slightly different language) gave a contrived proof to the following key statement.
\begin{quote}\vs{-1}
(*) For any index $d\geq2$, the language $L_d$ is located outside of $\cfl(d-1)$.
\end{quote}\vs{-1}
This result (*) then implies that the collection $\{\cfl(d)\mid d\geq1\}$ truly forms an infinite hierarchy.

\emph{Deterministic context-free (dcf) languages}, in contrast, have been another focal point of intensive research since an initiation of the systematic study in 1966 by Ginsburg and Greibach \cite{GG66}. The importance of these languages can be  exemplified by the facts that dcf languages are easy to parse and that every context-free language can be expressed simply as the homomorphic image of a dcf language.
Unlike $\cfl$, the family $\dcfl$ of all dcf languages is closed under neither the union nor the intersection operations.
Hereafter, we intend to use the terms of  \emph{$d$-intersection deterministic context-free (dcf) languages} and \emph{$d$-union deterministic context-free (dcf) languages} to respectively express any intersection of $d$ dcf languages and any union of $d$ dcf languages.
For brevity, we write $\dcfl(d)$ and $\dcfl[d]$ respectively for the \emph{$d$-intersection closure} and the \emph{$d$-union closure} of $\dcfl$, namely, the family of all $d$-intersection dcf languages and that of all $d$-union dcf languages.
As a special case, we obtain $\dcfl(1) = \dcfl[1] = \dcfl$. Since $\dcfl$ is closed under complementation, it also follows that the complement of $\dcfl(d)$ coincides with $\dcfl[d]$. For our convenience, we call the hierarchies $\{\dcfl(d)\mid d\geq1\}$ and $\{\dcfl[d]\mid d\geq1\}$ the \emph{intersection hierarchies of dcf languages} and the \emph{union hierarchies of dcf languages}, respectively.
The special notation $\dcfl(\omega)$ is meant to express the entire hierarchy  $\bigcup_{d\geq1} \dcfl(d)$, which is the intersection closure of $\dcfl$. Similarly, we write $\dcfl[\omega]$ for $\bigcup_{d\geq1} \dcfl[d]$.

In an analogy to the aforementioned union hierarchy of context-free languages, Wotschke \cite{Wot73,Wot78} discussed the union hierarchy $\{\dcfl[d]\mid d\geq1\}$ and asserted that this hierarchy truly forms an infinite hierarchy. He argued that, since the aforementioned language $L_d$ is in fact in $\dcfl(d)$, the result (*) of Liu and Weiner instantly yields $\dcfl(d)\nsubseteq \cfl(d-1)$, which further leads to the conclusion of $\dcfl(d-1)\neq \dcfl(d)$. Wotschke's argument, nonetheless, heavily relies on  Liu and Weiner's separation result (*), which uses a  specific property of \emph{stratified semi-linear sets}.
The proof technique of Liu and Weiner was developed specifically for $L_d$, which is a special form of \emph{bounded languages},\footnote{A \emph{bounded language} satisfies $L\subseteq w_1^*w_2^* \cdots w_k^*$ for fixed strings $w_1,w_2,\ldots,w_k$.} and it is therefore not applicable to ``arbitrary''  languages.
In fact, the key idea of their proof for $L_d$ is to focus on the number of the occurrences of each base symbol of   $\{a_1,\ldots,a_d,b_1,\ldots,b_d\}$ in a given string $w$ and to translate $L_d$ into a set $\Psi(L_d)$ of \emph{Parikh images} $(\#_{a_1}(w),\#_{a_2}(w),\ldots,\#_{a_d}(w),\#_{b_1}(w),\ldots,\#_{b_d}(w))$ in order to exploit the semi-linearity of $\Psi(L_d)$, where $\#_{\sigma}(w)$ expresses the total number of symbols $\sigma$ in $w$.

Because of the above-mentioned limitation of Liu and Weiner's proof technique, the scope of their proof cannot be directly extended to other forms of languages, such as the languages $L_d^{(\leq)} =\{a_1^{n_1}\cdots a_d^{n_d}b_1^{m_1}\cdots b_d^{m_d}\mid \forall i\in \{1,2,\ldots,d\}(n_i\leq m_i)\}$ and $NPal^{\#}_{d} = \{w_1\# w_2\# \cdots \# w_d \# v_1\# v_2 \#\cdots \# v_d \mid \forall i\in \{1,2,.\ldots,d\} (w_i,v_i\in\{0,1\}^*\wedge v_i\neq w_i^R)\}$, where $w_i^R$ expresses the \emph{reverse} of $w_i$.
The former is a bounded language expanding $L_d$ but its Parikh images do not have semi-linearity.  The latter is a ``non-palindrome'' language and it is not even a bounded language.
Unfortunately, Liu and Weiner's argument is not powerful enough to verify that neither $L^{(\leq)}_d$ nor $NPal^{\#}_d$ belongs to $\cfl(d-1)$ (unless we dextrously pick up its core strings to form a bounded language).
With no such contrived argument, how can we prove $L_d^{(\leq)}$ and $NPal_d^{\#}$ to be excluded from $\dcfl(d-1)$?
Furthermore, for any given language, how can we verify that it is outside of $\dcfl(\omega)$?
We can ask similar questions for $d$-union dcf languages and, more generally, for the union hierarchy $\dcfl[\omega]$ of dcf languages.
Using the special language $L_{wot}=\{wcx\mid w,x\in\{a,b\}^*,w\neq x\}$,  Wotschke \cite{Wot73,Wot78} proved that $L_{wot}$ does not belong to $\dcfl[\omega]$ (more strongly, the Boolean closure of $\dcfl$) by employing the closure property of $\dcfl(d)$ under inverse gsm mappings as well as intersection with regular languages. Wotschke's proof relies on the following two facts.
(i) The language $L_{d+1}$ can be expressed as the inverse gsm map of the language $Dup_c=\{wcw\mid w\in\{a,b\}^*\}$, restricted to $a_1^+a_2^+\cdots a_{d+1}^+ a_1^+a_2^+\cdots a_{d+1}^+$.
(ii) $Dup_c$ is expressed as the complement of $L_{wot}$, restricted to a certain regular language.
Together with these facts, the final conclusion comes from the aforementioned result (*) of Liu and Weiner because $Dup_{c}\in\dcfl(d)$ implies $L_{d+1}\in\dcfl(d)$ by (i) and (ii). To our surprise, all the fundamental results on $\dcfl(d)$ and $\dcfl[d]$ that we have discussed so far are merely ``corollaries'' of the main result (*) of Liu and Weiner!

In order to answer ``general'' non-membership questions to $\dcfl(d)$, we need to divert from Liu and Weiner's contrived argument used for proving the statement (*) and to develop a completely different, new, and more practical technical tool.
The sole purpose of this exposition is, therefore, set to  (i) develop a new proof technique, which can be applicable to many other languages, (ii) present an alternative proof for the fact that the intersection and the union hierarchies of DCFL are infinite hierarchies, and (iii) exhibit other languages in $\cfl$ that
do not belong to $\dcfl(\omega)\cup \dcfl[\omega]$.
In this direction of research, we refer to the work of Harrison and Havel \cite{HH74}.
Prior to Wotschke's work, Ginsburg and Greibach \cite{GG66} claimed \emph{with no proof}\footnote{In page 640 of \cite{GG66}, they wrote ``More strongly, $\{ ww^R \mid w \text{ in }\Sigma^*\}$ is not a finite union of deterministic
languages if $\Sigma$ contains at least two elements. We omit the proof.''} that the context-free language $Pal=\{ww^R\mid w\in\Sigma^*\}$ (even-length palindromes) for any non-unary alphabet $\Sigma$ is outside of  $\dcfl[\omega]$. This missing proof was later given by Harrison and Havel \cite{HH74} as  a quick application of their special ``iteration theorem'' for $\dcfl$.

In relevance to the union hierarchy of dcf languages, there is another known extension of $\dcfl$ within $\cfl$ using an access-controlled  machine model called \emph{limited automata},\footnote{Hibbard \cite{Hib67} actually defined a rewriting system, called ``scan-limited automata.'' Later, Pighizzini and Pisoni \cite{PP14,PP15} re-formulated Hibbard's system in terms of a restricted form of linear-bounded  automata.} which are originally invented by Hibbard \cite{Hib67} and later discussed extensively in, e.g., \cite{PP14,PP15,Yam19}.
A \emph{$d$-limited deterministic automaton} (or a $d$-lda, for short) is a one-tape deterministic Turing machine that can rewrite each tape cell located between two endmarkers only during the first $d$ visits (except that making a turn of a tape head counts as ``double visits''). We can raise a question of whether there is any relationship between the union hierarchy and $d$-lda's.

\subsection{New Separation Results}\label{sec:main-contribution}

As noted in Section \ref{sec:intersection-hierarchy}, some of the fundamental properties associated with $\dcfl(d)$ heavily rely on the single separation result (*) of Liu and Weiner. However, Liu and Weiner's technical tool that leads to their main result does not seem to withstand a wide variety of direct applications. It is thus desirable to  develop a new, simple, and practical technical tool that can find numerous applications for conducting a future study on $\dcfl(d)$ as well as $\dcfl[d]$.
Our main purpose of this exposition is therefore to present a simple but powerful and practical technical tool, which is applicable to numerous languages.
In a course of analyzing parse trees of (strict) dcf languages, Harrison and Havel \cite{HH74} developed an iteration theorem for $\dcfl$.
A similar approach has been taken to introduce numerous forms of pumping lemmas and iteration theorems for variants of dcf languages in,  e.g., \cite{Har86,Iga85,Kin80,KMW08,Wis76,Yu89}. Along this long line of study, this exposition also aims at developing \emph{pumping lemmas for $\dcfl[d]$} for any $d\geq1$ (formulated as Lemmas \ref{closure-REG} and \ref{hierarchy-properties} in Section \ref{sec:DFAs}), which may contribute to  the deepening of our understandings of both $\dcfl[d]$ and $\dcfl(d)$. See Fig.~\ref{fig:hierarchy-class}.


With the use of our pumping lemmas for $\dcfl[d]$, we significantly  expand the applicable scope of the argument of Liu and Weiner \cite{LW73}, which is  limited to specific bounded languages, to other types of languages, including the aforementioned languages $L_d^{(\leq)}$ and $NPal^{\#}_d$ for an arbitrary index $d\geq2$.

\begin{theorem}\label{two-langauge-proof}
For any integer $d$ at least $2$, the languages  $L_d^{(\leq)}$ and  $NPal^{\#}_d$ are not in $\dcfl(d-1)$.
\end{theorem}

From Theorem \ref{two-langauge-proof}, we instantly obtain the aforementioned consequence of Wotschke \cite{Wot73,Wot78}, stated in Corollary \ref{separation-proof}.
This statement directly comes from Theorem \ref{two-langauge-proof} because $\dcfl(d)$ contains both $L_d^{(\leq)}$ and $NPal_d^{\#}$.

\begin{corollary}\label{separation-proof}{\rm \cite{Wot73,Wot78}}
The intersection hierarchy of dcf languages and the union hierarchy of dcf languages are both infinite hierarchies.
\end{corollary}


\begin{figure}[t]
\centering
\includegraphics*[height=5.3cm]{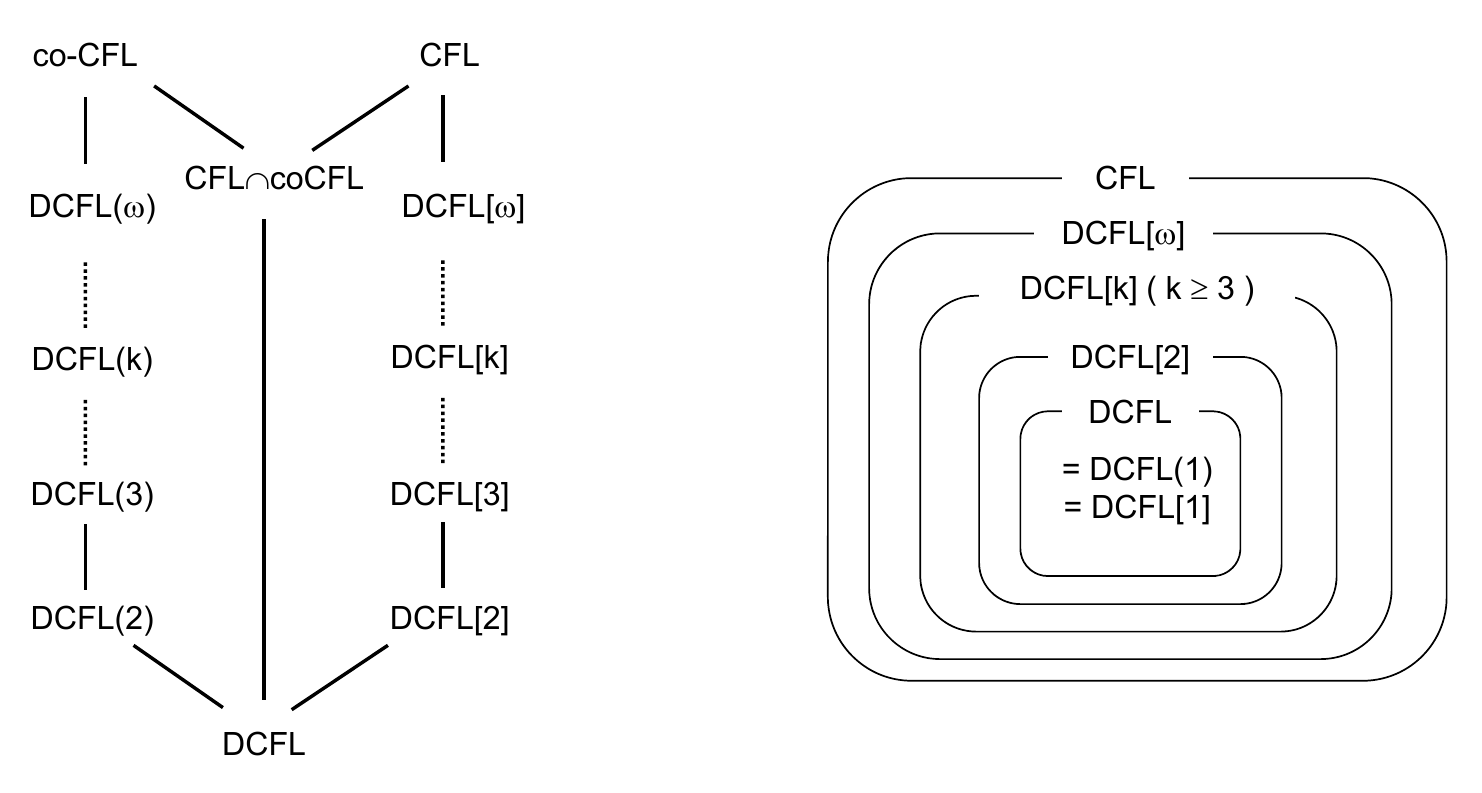}
\caption{Containment of language families. All families in this illustration are known to be distinct.}\label{fig:hierarchy-class}
\end{figure}


Concerning the limitation of $\dcfl(\omega)$ and $\dcfl[\omega]$ in recognition power, it is relatively easy to find highly complex languages that are not in $\dcfl[\omega]$. However, those languages may not serve themselves to separate $\cfl$ from $\dcfl(\omega)\cup\dcfl[\omega]$. As noted in Section \ref{sec:intersection-hierarchy}, the language $Pal$ of even-length palindromes
is known to be outside of $\dcfl[\omega]$ \cite{GG66,HH74}.
We can obtain the same result by way of
a direct application of our pumping lemmas.
Furthermore, we can show that another context-free language $MPal^{\#} =\{w_1\# w_2\# \cdots \# w_m \#^2  v_1\# v_2 \# \cdots \# v_n\mid m,n\geq1,
\exists i\in[\min\{m,n\}](|v_i|\neq |w_i|\vee v_i=w_i^R)\}$, where $w_1,w_2,\ldots,w_m,v_1,v_2,\ldots,v_n \in\{0,1\}^*$, is indeed outside of $\dcfl(\omega)\cup \dcfl[\omega]$.

\begin{theorem}\label{Pal-DCFL}
\renewcommand{\labelitemi}{$\circ$}
\begin{enumerate}
  \setlength{\topsep}{-2mm}%
  \setlength{\itemsep}{1mm}%
  \setlength{\parskip}{0cm}%

\item The language $Pal$ is in $\cfl\cap\co\cfl$ but not in $\dcfl[\omega]$.

\item The language $MPal^{\#}$ is in neither $\dcfl(\omega)$ nor $\dcfl[\omega]$.
\end{enumerate}
\end{theorem}

As an immediate consequence of Theorem \ref{Pal-DCFL}(1), we obtain Wotschke's separation result between $\cfl$ and $\dcfl[\omega]$. Here, we stress that, unlike the work of Wotschke \cite{Wot73,Wot78}, our proof does not depend on the main result (*) of Liu and Weiner. Actually, we can obtain a much stronger consequence (Corollary \ref{DCFL-omega}) from Theorem \ref{Pal-DCFL}(1) because $\cfl\cap\co\cfl\subseteq \dcfl(\omega)$ implies $\cfl\cap\co\cfl = \co(\cfl\cap\co\cfl) \subseteq \co\dcfl(\omega) = \dcfl[\omega]$, contradicting the theorem.

\begin{corollary}\label{DCFL-omega}
$\cfl\cap\co\cfl \nsubseteq \dcfl(\omega)\cup \dcfl[\omega]$.
\end{corollary}

As another application of our pumping lemmas,  we target languages associated with limited automata. For each index $d\geq1$, we write $\klda{d}$ for the family of all languages recognized by $d$-limited deterministic automata, in which their tape heads are allowed to rewrite tape symbols only during the first $d$ accesses (except that, in the case where tape heads make a turn, we treat each turn as double visits).
Hibbard demonstrated that $\klda{d}\neq \klda{(d-1)}$ for any integer $d\geq3$ \cite{Hib67}.
We expand this separation result to the following in connection to the intersection hierarchy of dcf languages.

\begin{proposition}\label{extension-Hibbard}
For any $d\geq2$, $\klda{d}\cap \dcfl[2^{d-1}] \nsubseteq \klda{(d-1)}\cup \dcfl[2^{d-1}-1]$.
\end{proposition}

The proofs of all of the above-mentioned assertions will be given in Section \ref{sec:proof-separation} after introducing necessary notions and notation in the subsequent section.

\section{Preparations: Notions and Notation}\label{sec:preparation}

We begin with a detailed explanation of basic notions and notation, which are crucial in the rest of this exposition.

\subsection{Fundamental Notions and Notation}\label{sec:notions-notation}

The set of all \emph{natural numbers} (including $0$) is denoted by $\nat$. An \emph{integer interval} $[m,n]_{\integer}$ for two integers $m$ and $n$ with $m\leq n$ is the set $\{m,m+1,m+2,\ldots,n\}$.
For any integer $n\geq1$, $[1,n]_{\integer}$ is succinctly abbreviated as $[n]$.
Given a set $S$, the notation $|S|$ indicates the \emph{cardinality} of $S$ and the notation $\PP(S)$ refers to the \emph{power set} of $S$.

An \emph{alphabet} is a finite nonempty set of ``symbols'' or ``letters''. A finite sequence of symbols in alphabet $\Sigma$ is called a \emph{string} over $\Sigma$. The \emph{length} of such a string $x$, denoted $|x|$, is the total number of symbols occurring in $x$. The special symbol $\varepsilon$ is used to denote the \emph{empty string} of length $0$. We write $\Sigma^*$ for the set of all strings over $\Sigma$ and we set $\Sigma^{+}$ to be $\Sigma^*-\{\varepsilon\}$. Given a number $n\in\nat$, the notation $\Sigma^n$ (resp., $\Sigma^{\leq n}$) further denotes the set of all strings of length exactly $n$ (resp., at most $n$) in $\Sigma^*$.
We say that $x$ is a \emph{substring} of $y$ (denoted $x\sqsubseteq y$) if $y=uxv$ holds for certain strings $u$ and $v$.
For a string $x$ of length $n$ and any index $i\in [n]$, the notation $x[i]$ expresses the $i$th symbol of $x$.
When a nonempty string $x$ is expressed as $x=x_1x_2\cdots x_n$ with strings $x_1,x_2,\ldots,x_n$, we call such an expression a \emph{factorization} of $x$. For clarity, we also write $(x_1,x_2,\ldots,x_n)$.

Any subset of $\Sigma^*$ is called a \emph{language} over $\Sigma$.
For a language $L$ over $\Sigma$,  $\Sigma^*-L$ is the \emph{complement} of $L$ and expressed as $\overline{L}$ as long as $\Sigma$ is clear from the context. Given a family $\FF$ of languages, $\co\FF$ expresses the complement family, which consists of the complements $\overline{L}$ of all languages $L\in\FF$.

For two language families $\FF_1$ and $\FF_2$, the notation $\FF_1\wedge \FF_2$ (resp., $\FF_1\vee\FF_2$) denotes the family of all languages $L$  such that there are two languages $L_1\in\FF_1$ and $L_2\in\FF_2$ over a certain common alphabet satisfying $L=L_1\cap L_2$ (resp., $L=L_1\cup L_2$). Generally, for a given  $k$-ary operation, say, $op$ over $k$ languages, we say that a family $\CC$ of languages \emph{is closed under} $op$ if, for any $k$ languages in $\CC$, $op(L_1,L_2,\ldots,L_k)$ also belongs to $\CC$.

\subsection{Deterministic Pushdown Automata}\label{sec:DFAs}

The machine model of \emph{one-way deterministic pushdown automata} (abbreviated as 1dpda's) was introduced in 1966 by Ginsburg and Greibach \cite{GG66}. Formally, a 1dpda $M$ is a nonuple $(Q,\Sigma,\{\cent,\dollar\},\Gamma,\delta,q_0,Z_0,Q_{acc},Q_{rej})$, where $Q$ is a finite set of inner states, $\Sigma$ is an input alphabet, $\Gamma$ is a stack alphabet, $\delta$ is a transition function from $(Q-Q_{halt})\times\check{\Sigma}_{\varepsilon} \times\Gamma$ to $\PP(Q\times\Gamma^*)$ with $Q_{halt} = Q_{acc}\cup Q_{rej}$, $\check{\Sigma} = \Sigma\cup\{\cent,\dollar\}$, and $\check{\Sigma}_{\varepsilon} = \check{\Sigma}\cup\{\varepsilon\}$ satisfying $|\delta(p,\sigma,a)|\leq 1$ for any $(p,\sigma,a)$, $q_0$ is the initial state in $Q$, $Z_0$ is the bottom marker in $\Gamma$, and $Q_{acc}$ and $Q_{rej}$ are respectively sets of accepting states and of rejecting states.
Remember that $\cent$ and $\dollar$ are respectively called the left-endmarker and the right-endmarker. Let $\Gamma^{(-)} = \Gamma-\{Z_0\}$.
Any inner state in $Q_{halt}$ is called a \emph{halting state}.
The \emph{push size} of a 1ppda is the maximum length of any string pushed into a stack by any single move.
The 1dpda $M$ must satisfy the following \emph{deterministic requirement}:  $|\delta(p,\sigma,a)\cup \delta(p,\varepsilon,a)|=1$ holds for any $p\in Q$, any $a\in\Gamma$, and any symbol $\sigma\in\check{\Sigma}$.
When $\delta(q,\sigma,a)$ is a singleton, say, $\{(p,w)\}$, we intentionally write ``$\delta(p,\sigma,a)=(p,w)$''. This transition indicates that, when $M$ is in inner state $q$ scanning $\sigma$ on an input tape and $a$ on the top of a stack, $M$ changes $q$ to $p$, replaces $a$ by $w$, and moves its tape head to the right whenever $\sigma\neq\varepsilon$.
Moreover, we require that the bottom marker $Z_0$ is not removable or replaceable at any time, and $Z_0$ is not allowed to push on top of any stack symbol.
A content of the stack (or a \emph{stack content}, for short) is expressed as $a_1a_2\cdots a_k$ so that $a_1$ is the topmost stack symbol, $a_k$ is the bottom marker, and all $a_i$'s are placed in order from the top to the bottom in the stack.
The \emph{stack height} refers to the size of a stack content, namely, the total number of symbols stored in the stack.
For instance, a stack content $a_1a_2\cdots a_k$ with $a_k=Z_0$ has stack height $k$. Since $Z_0$ is not removable, the stack height is always at least $1$ during any computation. A consecutive series of the transitions of stack contents produced by $M$ is
generally referred to as a \emph{stack history}.

When $M$ enters a halting state, $M$ \emph{halts} (by the definition of $\delta$). In the rest of this exposition, we always demand that $M$ halts on all inputs. Given a string $w$, we say that $M$ \emph{accepts} (resp., \emph{rejects}) $w$ if $M$ is in an accepting (resp., rejecting) state when it halts. We use the notation $L(M)$ to denote the set of all strings accepted by $M$. If a language $L$ satisfies $L=L(M)$, then $M$ is said to \emph{recognize} $L$. Such a language is called a \emph{deterministic context-free (dcf) language}.

Given any number $d\in\nat^{+}$, a \emph{$d$-intersection deterministic context-free (dcf) language} refers to an intersection of $d$ dcf  languages. Let $\dcfl(d)$ denote the family of all such $d$-intersection dcf languages. Similarly, we define \emph{$d$-union dcf languages} and $\dcfl[d]$ by substituting ``union'' for  ``intersection'' in the above definition.
Remark that $\dcfl = \dcfl(1)=
\dcfl[1]$. It then follows that $\dcfl[d] = \co(\dcfl(d))$ because of $\dcfl=\co\dcfl$.

The following two lemmas will be quite useful in proving Theorems \ref{two-langauge-proof} and \ref{Pal-DCFL} together with Proposition \ref{extension-Hibbard}.

\begin{lemma}\label{closure-REG}{\rm \cite{Wot73,Wot78}}
$\dcfl(d)$ is closed under union and intersection with regular languages.  In other words, $\dcfl(d)\wedge \reg \subseteq \dcfl(d)$ and $\dcfl(d)\vee \reg \subseteq \dcfl(d)$. A similar statement holds for $\dcfl[d]$.
\end{lemma}

\begin{proof}
For completeness, we provide the proof of the lemma. Let us take any language $L\in\dcfl[d]$ and a regular language $A$. For this $L$, take $d$ appropriate languages $L_1,L_2,\ldots,L_d\in \dcfl$ satisfying  $L=\bigcap_{i\in[d]} L_i$. It then follows that $L\cap A=\bigcap_{i\in[d]}(L_i\cap A)$ and $L\cup A=\bigcap_{i\in[d]}(L_i\cup A)$. Since $L_i\cap A$ and $L_i\cup A$ are both deterministic context-free, we conclude that $L\cap A$ and $L\cup A$ both belong to $\dcfl(d)$. The case for $\dcfl[d]$ is similarly treated.
\end{proof}

\begin{lemma}\label{hierarchy-properties}
Let $d\geq1$ denote any integer.
\renewcommand{\labelitemi}{$\circ$}
\begin{enumerate}\vs{-1}
  \setlength{\topsep}{-2mm}%
  \setlength{\itemsep}{1mm}%
  \setlength{\parskip}{0cm}%

\item $\dcfl(d)=\dcfl(d+1)$ iff $\dcfl[d]=\dcfl[d+1]$.

\item For any language $L\in\dcfl(d)$, it follows that $A\cap\overline{L}\in\dcfl[d]$ for any language $A\in\reg$.
\end{enumerate}
\end{lemma}

\begin{proof}
(1) Assume that $\dcfl(d)=\dcfl(d+1)$. By taking complementation, we obtain  $\co(\dcfl(d)) = \co(\dcfl(d+1))$. This is equivalent to $\dcfl[d]=\dcfl[d+1]$. The other direction is similarly proven.

(2) Assume that $L\in\dcfl(d)$. From this follows $\overline{L}\in\dcfl[d]$. Lemma \ref{closure-REG} then yields $A\cap\overline{L}\in\dcfl[d]$ for any regular language $A$.
\end{proof}

From Lemma \ref{hierarchy-properties}(1) follows Corollary  \ref{separation-proof}, provided that Theorem \ref{two-langauge-proof} is true. Theorem \ref{two-langauge-proof} itself will be proven in Section \ref{sec:proof-separation}.

\subsection{Pushdown Automata in an Ideal Shape}\label{sec:ideal-shape}

We will introduce two pumping lemmas for $\dcfl[d]$ as Lemmas \ref{pumping-lemma} and \ref{second-pumping-lemma} in Section \ref{sec:new-tool} and prove them in Section \ref{sec:proof-pumping}. To make these proofs simpler, we intend to exploit the fact that any 1dpda can be converted into a specific form.
Let us recall from \cite{Yam19,Yam21} a special ``push-pop-controlled''  form (called an \emph{ideal shape}), in which (i) a pop operation always takes  place by first reading an input symbol and then making  a series (one or more) of the pop operations without reading any further input symbol and
(ii) a push operation always adds a single symbol
without altering any existing stack content.
The original notion given in \cite{Yam19,Yam21} was meant for \emph{one-way probabilistic pushdown automata} (or 1ppda's); however, in this exposition, we apply this notion to 1dpda's.
To be more precise, a 1ppda \emph{in an ideal shape} is a 1ppda restricted to the following moves regarding its stack operations.
(1) Scanning $\sigma\in\Sigma$, preserve the topmost stack symbol (called a \emph{stationary operation}).  (2) Scanning  $\sigma\in\Sigma$,  push a new symbol $u$ ($\in\Gamma^{(-)}$) without changing any other symbol in the stack. (3) Scanning  $\sigma\in\Sigma$, pop the topmost stack symbol. (4) Without scanning an input symbol (i.e., $\varepsilon$-move), pop the topmost stack symbol. (5) The stack operation (4) comes only after either (3) or (4).

It was shown in \cite{Yam19,Yam21} that any 1ppda can be converted into its ``equivalent'' 1ppda in an ideal shape.  We say that two 1ppda's are \emph{error-equivalent} if, for any input $x$, their acceptance probabilities coincide.

\begin{lemma}{\rm [Ideal Shape Lemma for 1ppda's \cite{Yam19,Yam21}]}\label{ideal-shape}
Let $n\in\nat^{+}$. Any $n$-state 1ppda $M$ with stack alphabet size $m$ and push size $e$ can be converted into another error-equivalent 1ppda $N$ in an ideal shape with $O(en^2m^2(2m)^{2enm})$ states and stack alphabet size $O(enm(2m)^{2enm})$.
\end{lemma}

As noted in \cite{Yam19,Yam21}, by setting an error probability of 1ppda's to be $0$, Lemma \ref{ideal-shape} becomes applicable to 1dpda's. Therefore, it suffices in the rest of this exposition to deal only with 1dpda's in an ideal shape.

\subsection{Deterministic Limited Automata}

Hibbard \cite{Hib67} discussed another machine model of scan limited automata in 1967. In this exposition, we follow a reformulation of his model by Pighizzini and Pisoni \cite{PP14,PP15} and Yamakami \cite{Yam19} under the name of \emph{deterministic $d$-limited automata} (or $d$-lda's, for short). Given $d\geq1$, a $d$-lda $M$ is formally an octuple $(Q,\Sigma,\{\cent,\dollar\}, \{\Gamma^{(e)}\}_{e\in[d]}, \delta,q_0, Q_{acc},Q_{rej})$ with a transition function $\delta$ mapping $(Q-Q_{halt})\times \check{\Sigma}\times \Gamma$ to $Q\times \Gamma \times D$, where $D=\{-1,+1\}$,  $\Gamma=\bigcup_{e\in[0,d]_{\integer}} \Gamma^{(e)}$ with $\Gamma^{(0)}=\Sigma$, $\{\cent,\dollar\}\subseteq\Gamma^{(d)}$, and $\Gamma^{(i)}\cap \Gamma^{(j)}=\setempty$ for any distinct pair $i,j\in[0,d]_{\integer}$.
We remark that $M$ has a single rewritable input/work tape and moves its tape head in both directions without making any $\varepsilon$-move. If $\delta(q,\sigma) = (p,\tau,l)$, then $M$ changes its inner state $q$ to $p$, writes $\tau$ over $\sigma$, and moves its tape head in direction $l$.
Assuming further that $\sigma\in\Gamma^{(i)}$ and $\tau\in\Gamma^{(j)}$ for $i,j\in[0,d]_{\integer}$, we demand that (1) if $i=d$, then $\sigma=\tau$ and $j=d$, (2) if $i<d$ and $i$ is even, then $j=i+2^{(1-l)/2}$, and (3) if $i<d$ and $i$ is odd, then $j=i+2^{(1+l)/2}$. These requirements imply that no symbol in $\Gamma^{(d)}$ is replaced by any other symbol.

Here, we reuse the basic terminology introduced in Section \ref{sec:DFAs} for 1dpda's and we do not intend to restate similar definitions for $d$-lda's.
The notation $\klda{d}$ is used for the family of all languages recognized by $d$-lda's. It is known that $\dcfl =\klda{2}$ \cite{PP15} and $\klda{d}\neq \klda{(d+1)}$ for any index $d\in\nat^{+}$ \cite{Hib67}.

\section{New Technical Tools: Pumping Lemmas for DCFL[$d$]}\label{sec:new-tool}

In the past literature, numerous pumping lemmas (as well as iteration theorems) have been proposed. Interestingly, the form of these pumping lemmas varies. For instance, the iteration theorem  of Harrison and Havel \cite[Theorem 2.3]{HH74} for $\dcfl$ focuses on a single input string in a target language, whereas Yu's pumping lemma for $\dcfl$ \cite[Lemma 1]{Yu89} deals with two input strings in the language. We wish to take the latter approach in this exposition.

A pumping lemma generally asserts the existence of a pair of ``repeatable'' portions in a given string. Such a pair is known as an \emph{iterative pair} \cite{Boa73}. Given a language $L$ and a string $w$ of the form $uxvyz$ (called a \emph{factorization}), a pair $(x,y)$ is called an \emph{iterative pair}\footnote{The original notion in \cite{Boa73} does not include the case of $i=0$. This current notion appeared in Berstel's textbook \cite{Ber79} and it is also called ``strong iterative pair'' in \cite{ER85}. The notation $(x,y)$ used here, instead of $(u,x,v,y,z)$, is adopted from Section VIII.4 of Berstel's textbook \cite{Ber79}, in which a similar succinct notation was used to denote an iterative pair.}
of $w$ for $L$ if $|xy|\geq1$ and $ux^ivy^iz\in L$ for any nonnegative integer $i$.

We describe the first pumping lemma for $\dcfl[d]$ using the notion of iterative pairs.

\begin{lemma}{\rm [The First Pumping Lemma for {DCFL[$d$]}]}\label{pumping-lemma}
Let $d$ be any positive integer and let $L$ denote any infinite $d$-union dcf language. There exist a constant $c>0$ such that, for any $d+1$ strings $w_1,w_2,\ldots,w_{d+1}\in L$,
if $w_i$ has the form $xy^{(i)}$ with $|x|>c$ and $|y^{(i)}|\geq1$ for any index $i\in[d+1]$, then there exist two distinct indices $j_1,j_2\in[d+1]$ for which the following statement holds. For any strings $x'$, $y$, and $z$ with $|x'|>c$, if $x'y = xy^{(j_1)}$ and $x'z=xy^{(j_2)}$, then one of the following conditions (1)--(2) must hold.
\renewcommand{\labelitemi}{$\circ$}
\begin{enumerate}\vs{-2}
  \setlength{\topsep}{-2mm}%
  \setlength{\itemsep}{1mm}%
  \setlength{\parskip}{0cm}%

\item There is a factorization $x' = x_1x_2x_3x_4x_5$ with $|x_2x_4|\geq1$ and $|x_2x_3x_4|\leq c$ such that $(x_2,x_4)$ is an iterative pair of both $x'y$ and $x'z$ for $L$.

\item There are three factorizations $x' = x_1x_2x_3$, $y = y_1y_2y_3$, and $z = z_1z_2z_3$ with $|x_2|\geq1$ and $|x_2x_3|\leq c$ such that $(x_2,y_2)$ and $(x_2,z_2)$ are respectively iterative pairs of $x'y$ and $x'z$ for $L$.
\end{enumerate}
\end{lemma}

As a special case of $d=1$ in Lemma \ref{pumping-lemma}, we instantly obtain Yu's pumping lemma for $\dcfl$ \cite[Lemma 1]{Yu89}, which was proven by a \emph{grammar-based argument}.
By sharp contrast, a few \emph{machine-based arguments} were employed in the past literature in order to prove generic properties of $\cfl$ and $\dcfl$. As a recent example, the \emph{swapping lemma for context-free languages}  was proven in \cite[Lemma 4.1]{Yam08}
and \cite[Corollary 4.2]{Yam16} by exploring various properties of one-way nondeterministic pushdown automata (or 1npda's).
Such machine-based arguments as well as various properties of machines that these arguments have brought in are not only interesting in its own right but also quite useful to analyze families of languages recognized by different types of machines, such as one-way depth-bounded storage automata \cite{Yam21b,Yam22,Yam25}, which  naturally expand 1npda's and 1dpda's but no grammar characterizations of them are currently known to exploit.

At first glance, Yu's grammar-based argument for his pumping lemma looks quite simple but in fact it heavily relies on his so-called \emph{left-part theorem for LR($k$) grammars} whose proof (not available for general public) seems quite lengthy (by measuring the size of its alternative proof given in Section \ref{sec:grammatical-tree}). In the end, a machine-based argument (presented in Section \ref{sec:lemma-verify}) looks no more complex than his grammar-based argument.
From this perspective, one of the important aspects of this exposition is a clear demonstration of how to explore new behavioral properties of 1dpda's in hopes of developing new analytical techniques that can be extendable to  other types of machines.
Here, we intend to restate Yu's lemma in a slightly stricter form using $n$-state 1dpda's as Lemma \ref{pumping-lemma-state}. The notion of an \emph{ideal shape form} in the lemma, which comes from \cite{Yam19,Yam21}, is explained in Section \ref{sec:ideal-shape}.

\begin{lemma}\label{pumping-lemma-state}
Let $M$ denote any $n$-state 1dpda in an ideal shape. For any two strings $xy$ and $xz$ in $L(M)$ with $|x|> 2^{6n^6}$, one of the following conditions (1)--(2) must hold.
\renewcommand{\labelitemi}{$\circ$}
\begin{enumerate}\vs{-2}
  \setlength{\topsep}{-2mm}%
  \setlength{\itemsep}{1mm}%
  \setlength{\parskip}{0cm}%

\item There is a factorization $x = x_1x_2x_3x_4x_5$ with $|x_2x_4|\geq1$ and $|x_2x_3x_4|\leq 2^{6n^6}$ such that $(x_2,x_4)$ is an iterative pair of both $xy$ and $xz$ for $L(M)$.

\item There are three factorizations $x = x_1x_2x_3$, $y = y_1y_2y_3$, and $z = z_1z_2z_3$ with $|x_2|\geq1$ and $|x_2x_3|\leq 2^{6n^6}$ such that $(x_2,y_2)$ and $(x_2,z_2)$ are respectively iterative pairs of $xy$ and $xz$ for $L(M)$.
\end{enumerate}
\end{lemma}


The proof of Lemma \ref{pumping-lemma-state} also exploits an early result of \cite{Yam19,Yam21} on an ideal shape form together with an approach with \emph{$\varepsilon$-enhanced machines} by analyzing transitions of \emph{state-stack pairs} in Section \ref{sec:state-stack}. These additional notions will be explained in Section \ref{sec:proof-pumping} and their basic properties will be explored therein.


Unfortunately, Lemma \ref{pumping-lemma} is so restrictive that it cannot be directly applied to the verification of, for example, Theorem \ref{two-langauge-proof}(2). We therefore need a more powerful pumping lemma to prove it.

An iterative pair $(x,y)$ is further said to be \emph{nondegenerate} if either $\{ux^ivy^jz\in L \mid j\geq0\}$ is finite for every $i\geq0$ or $\{ux^ivy^jz\in L \mid i\geq0\}$ is finite for every $j\geq0$. Otherwise, $(x,y)$ is said to be \emph{degenerate}. When either $x$ or $y$ is empty, the iterative pair $(x,y)$ is said to be \emph{empty}. Every nondegenerate iterative pair is obviously nonempty.
For instance, the language $L=\{ab^nca^mb\mid m,n\in\nat\}$ has an iterative pair $(b,a)$ but it is obviously degenerate.
The nondegenerate iterative pairs for dcf languages have been discussed in the past literature (e.g., \cite{Rub18,Sen90}).

The second pumping lemma for $\dcfl[d]$ is described as follows using nondegenerate iterative pairs.

\begin{lemma}{\rm [The Second Pumping Lemma for {DCFL[$d$]}]}\label{second-pumping-lemma}
Let $d$ be any positive integer and let $L$ denote any infinite $d$-union dcf language. There exists a constant $c>0$ that satisfies the following statements. For arbitrary $d+1$ strings $w_1,w_2,\ldots,w_{d+1}\in L$,
if $w_i$ has the form $xy^{(i)}$ with $|x|,|y^{(i)}|\geq1$ for any index $i\in[d+1]$, then
there exist two distinct indices $j_1,j_2\in[d+1]$ such that, for any three nonempty strings $x'$, $y$, and $z$, if $x'y = xy^{(j_1)}$ and $x'z=xy^{(j_2)}$ with $|x'|>c$, then one of the following conditions (1)--(5) holds.
\renewcommand{\labelitemi}{$\circ$}
\begin{enumerate}\vs{-2}
  \setlength{\topsep}{-2mm}%
  \setlength{\itemsep}{1mm}%
  \setlength{\parskip}{0cm}%

\item There exists a factorization $x'=x_1x_2x_3x_4x_5$ with $|x_2x_4|\geq1$ and $|x_2x_3x_4|\leq c$ such that $(x_2,x_4)$  is an iterative pair of $x'u$ of $L$ for any string $u\in\{y,z\}$.

\item There exist a string $u\in\{y,z\}$ and its factorization $u=u_1u_2u_3u_4u_5$ with $|u_2u_4|\geq1$ and $|u_2u_3u_4|\leq c$ such that $(u_2,u_4)$ is an iterative pair of $x'u$ of $L$.

\item There are four factorizations $x'=x_1x_2x_3 = x'_1x'_2x'_3$, $y=y_1y_2y_3$, and $z=z_1z_2z_3$ with $1\leq |x_2|,|x'_2|\leq c$ such that $(x_2,y_2)$ and $(x'_2,z_2)$ are respectively nondegenerate iterative pairs of $x'y$ and $x'z$ for $L$.

\item There are four factorizations $x'=x_1x_2x_3x_4 = x'_1x'_2x'_3x'_4$, $y=y_1y_2y_3$, and $z=z_1z_2z_3$ with $|x_4|,|x'_4|,|y_1|,|z_1|\geq1$ and $|x_2|,|x'_2|\leq c$ such that $(x_2,x_4y_1)$ and $(x'_2,x'_4z_1)$ are respectively nondegenerate iterative pairs of $x'y$ and $x'z$ for $L$.

\item There exist a string $u\in\{y,z\}$ and three factorizations $x'=x_1x_2x_3x_4x_5x_6x_7$, $u=u_1u_2u_3$, and $u^{(op)}=u'_1u'_2u'_3$ with $|x_7|\geq1$ and $|x_4x_5|\leq c$ such that (i) $(x_4x_5,x_7u_1)$ is a nondegenerate iterative pair of $x'u$ for $L$ and
    (ii) either $(x_2,u'_2)$ with $|x_2|\leq c$ or $(x_3x_4,u'_2)$ with $|x_3x_4|\leq c$ is a nondegenerate iterative pair of $x'u^{(op)}$ for $L$, where $u^{(op)}$ denotes a unique element in $\{y,z\}-\{u\}$.
\end{enumerate}
\end{lemma}

This second pumping lemma for $\dcfl[d]$ (Lemma \ref{second-pumping-lemma}) was originally proven by way of a machine-based argument in the preliminary conference version \cite{Yam20}. For the reader who are familiar with the grammar characteristics of dcf languages, we dare to employ a grammar-based argument for the proof of the lemma in Section \ref{sec:second-lemma-verify}.

\section{Proofs of Three Separation Claims}\label{sec:proof-separation}

The proofs of the two pumping lemmas for $\dcfl[d]$ (i.e., Lemmas \ref{pumping-lemma} and \ref{second-pumping-lemma}) require a lengthy argument and we postpone the proofs of the lemmas until
Sections \ref{sec:proof-pumping} and \ref{sec:second-lemma-verify}.
Meanwhile, we intend to concentrate on the three separation claims (Theorems \ref{two-langauge-proof} and \ref{Pal-DCFL} and Proposition \ref{extension-Hibbard}) announced in Section \ref{sec:main-contribution} and provide their detailed proofs throughout this section.
To understand these proofs better, we first demonstrate a simple and easy example of how to apply the first pumping lemma for $\dcfl[d]$ (Lemma \ref{pumping-lemma}) to obtain the desired separation between $\dcfl[d]$
and $\dcfl[d-1]$.


\begin{proposition}
Let $d\geq2$ and set $L_{(d)}=\{a^nb^{kn}\mid k\in[d],n\geq0\}$. It then follows that $L_{(d)}\in \dcfl[d] - \dcfl[d-1]$.
\end{proposition}

\begin{proof}
Let $d\geq2$. Since $L_{(d)}$ can be expressed as $\bigcup_{k\in[d]} L^{(k)}$ using the $d$ languages $L^{(k)}=\{a^nb^{kn}\mid n\geq0\}$ for any fixed number $k\in[d]$, $L_{(d)}$ clearly belongs to $\dcfl[d]$.
Toward the non-membership $L_{(d)}\notin \dcfl[d-1]$, we assume otherwise and
apply to $L_{(d)}$ the first pumping lemma for $\dcfl[d]$.
Take a pumping-lemma constant $c>0$ that satisfies the lemma. Let $n=c+1$ and consider $x=a^n$ and $y^{(i)}=b^{in}$ for each index $i\in[d]$. Since each string $w_i= a^nb^{in}$ belongs to $L_{(d)}$, there is a special index pair $j,k\in[d]$ with $j<k$ for which $w_j$ and $w_k$ satisfy the lemma.

Let us consider the condition (1) of the pumping lemma. Let $x'=a^nb^{jn-1}$, $y=b$, and $z=b^{(k-j)n+1}$. Firstly, let us focus on a factorization $x'=x_1x_2x_3x_4x_5$ with $|x_2x_4|\geq1$ and $|x_2x_3x_4|\leq c$. Since $x_1x_2^ix_3x_4^ix_5y\in L_{(d)}$ holds for any $i\in\nat$, it follows that $x_2\in\{a\}^*$ and $x_4\in\{b\}^*$.
Because of the definition of $L_{(d)}$, we also conclude that $x_2\neq\varepsilon$ and $x_4\neq\varepsilon$.
Since $|x_2x_3x_4|\leq c$, we can assume that $x_2=a^m$ and $x_4 = b^r$ for certain numbers $m,r\in[c]$.
Note that, for any $i\in\nat$,  $x_1x_2^ix_3x_4^ix_5y$ equals $a^{n+(i-1)m}b^{jn+(i-1)r}$.
Since this string is in $L_{(d)}$,  $jn+(i-1)r = g(n+(i-1)m)$ holds for an appropriately chosen number $g\in[d]$.
This implies that $(g-j)n = (r-mg)(i-1)$. We fix such a number $g$ so that this equality holds for infinitely many $i\in\nat$. We then obtain $g=j$ and $r=mg$. From $x_1x_2^ix_3x_4^ix_5z\in L_{(d)}$, we observe that, for each $i\in\nat$, $kn+(i-1)r = g'(n+(i-1)m)$ holds for a certain number $g'\in[d]$. We choose such a number $g'$ satisfying $(g'-k)n=(r-mg')(i-1)$ for infinitely many $i\in\nat$. It thus follows that   $g'=k$ and $r=mg'$. As a result, we obtain $r=mj=mk$, which implies $j=k$. Since $j\neq k$, we obtain a clear contradiction.

Secondly, we assume that the condition (1) fails. The condition (2) then provides three factorizations $x' =x_1x_2x_3$, $y=y_1y_2y_3$, and $z=z_1z_2z_3$ with $|x_2|\geq1$ and $|x_2x_3|\leq c$ such that $x_1x_2^ix_3 y_1y_2^i y_3\in L_{(d)}$ and $x_1x_2^ix_3 z_1z_2^i z_3\in L_{(d)}$ hold for any number $i\in\nat$. Since $|x_2x_3|\leq c$, we obtain $x_2\in\{b\}^{+}$. Assume that $x_2=b^{m}$ for a certain number  $m\in[c]$.
Since the condition (1) fails, we obtain $y_2\neq \varepsilon$, and thus $y_2=b$ follows. Since $x_1x_2^ix_3y_1y_2^iy_3$ has the form $a^nb^{jn+(i-1)(m+1)}$, it follows that, for any number $i\in\nat$, there exists a number $g\in[d]$ satisfying $jn+(i-1)(m+1) = gn$, which implies $(g-j)n=(i-1)(m+1)$. We fix a number $g$ satisfying this equality for infinitely many numbers $i$. However, this is impossible because the right-hand side of the equality has infinitely many values.
Therefore, $L_{(d)}$ is not in $\dcfl[d-1]$.
\end{proof}

Let us recall the languages $L_d^{(\leq)} =\{a_1^{n_1}\cdots a_d^{n_d}b_1^{m_1}\cdots b_d^{m_d}\mid \forall i\in[d](n_i\leq m_i)\}$ over the alphabet $\Sigma_d=\{a_1,a_2,\ldots,a_d, b_1,b_2,\ldots,b_d\}$ and $NPal^{\#}_d = \{w_1\# \cdots \# w_d \# v_1 \# \cdots \# v_d \mid \forall i\in[d] (w_i,v_i\in\{0,1\}^* \wedge v_i\neq w_i^R)\}$ over the alphabet $\{0,1,\#\}$. Theorem \ref{two-langauge-proof} asserts their  non-membership to $\dcfl(d-1)$.
Since our pumping lemmas are concerned with $\dcfl[d]$, we first need to take the complements of those languages, restricted to suitable regular languages, and we then apply the pumping lemmas to verify the theorem.

\ms
\n{\bf Theorem \ref{two-langauge-proof}.} (rephrased)
{\it Let $d\geq2$.
\renewcommand{\labelitemi}{$\circ$}
\begin{enumerate}\vs{-2}
  \setlength{\topsep}{-2mm}%
  \setlength{\itemsep}{1mm}%
  \setlength{\parskip}{0cm}%

\item The language $L_d^{(\leq)}$ is not in $\dcfl(d-1)$.

\item The language $NPal^{\#}_d$ is not in $\dcfl(d-1)$.
\end{enumerate}}

\begin{proof}
Let $d\geq2$ be any positive integer.

(1)
Our first target is the language $L^{(\leq)}_{d}$ over the alphabet $\Sigma_d$. For each index $i\in[d]$, we set $L^{(i)}$ to be $\{a_1^{n_1}\cdots a_d^{n_d}b_1^{m_1}\cdots b_d^{m_d}\mid 0\leq n_i \leq m_i\}$. It is clear that $L_d^{(\leq)}$ coincides with $\bigcap_{i\in[d]}L^{(i)}$, and therefore $L_d^{(\leq)}$ belongs to $\dcfl(d)$.
Our goal here is to verify the non-membership of $L_d^{(\leq)}$ to $\dcfl(d-1)$. To lead to a contradiction, let us assume that $L_d^{(\leq)} \in\dcfl(d-1)$.
Take the special regular language $A=a_1^*a_2^*\cdots a_d^*b_1^*b_2^*\cdots b_d^*$ and define $L'=A\cap (\Sigma_d^* - L_d^{(\leq)})$; in other words, $L'=\{a_1^{n_1}\cdots a_d^{n_d}b_1^{m_1}\cdots b_d^{m_d}\mid \exists i\in[d](n_i > m_i\geq0)\}$.
Note by Lemma \ref{hierarchy-properties}(2) that, since $L^{(\leq)}_d \in\dcfl(d-1)$, $L'$ belongs to $\dcfl[d-1]$.
This makes it possible to apply to this $L'$ the first pumping lemma for $\dcfl[d-1]$ (Lemma \ref{pumping-lemma}).

We take a pumping-lemma constant $c>0$ that satisfies the pumping lemma, and we then set  $n=c+1$. Let us focus on $d$ strings $xy^{(1)},xy^{(2)},\cdots, xy^{(d)}$, where
$x = a_1^{n}a_2^{2n} \cdots a_d^{dn}$ and $y^{(k)} = b_1^{n}b_2^{2n} \cdots b_{k-1}^{(k-1)n} b_{k}^{kn-1} b_{k+1}^{(k+1)n} \cdots b_d^{dn}$ for any number $k\in[d]$. In other words, $y^{(k)}$ is similar to $x$ except that the $k$th block $b_k^{kn-1}$ is different from $b_k^{kn}$.
The pumping lemma then guarantees the existence of a distinct pair $j_1,j_2\in [d+1]$ with $j_1<j_2$ satisfying
the lemma's conditions (1)--(2).
Let $x'= x b_1^{n}\cdots b_{j_1-1}^{(j_1-1)n} b_{j_1}^{j_1n-2}$, $y =  b_{j_1} b_{j_1+1}^{(j_1+1)n} \cdots b_d^{dn}$, and $z = b_{j_1+1}^{(j_1+1)n} \cdots b_{j_2-1}^{(j_2-1)n} b_{j_2}^{j_2n-1} b_{j_2+1}^{(j_2+1)n} \cdots b_d^{dn}$. Since $|x'|>c$, it then follows that $x'y = xy^{(j_1)}$ and $x'z = xy^{(j_2)}$.

Let us examine each condition of the pumping lemma.
Firstly, we consider the condition (1), in which there is a factorization $x'=x_1x_2x_3x_4x_5$ with $|x_2x_4|\geq1$ and $|x_2x_3x_4|\leq c$ satisfying $x_1x_2^ix_3x_4^ix_5u\in L'$ for any $u\in\{y,z\}$ and any number $i\in\nat$.
By the definition of $L'$, $x_2$ cannot contain a substring of the form $a_db_1$ as well as $a_ja_{j+1}$ for a certain index $j\in[d-1]$. A similar claim holds for $x_4$. Assuming $x_2\neq\varepsilon$, we take a number $k\in[d]$ for which $x_2\in\{a_k\}^{+}$. This implies that $x_4$ must contain $b_k$. However, this is impossible because of $|x_2x_3x_4|\leq c$.  The case of $x_4\neq\varepsilon$ is similarly handled.

Next, we consider the condition (2), in which there are three factorizations $x'= x_1x_2x_3$, $y=y_1y_2y_3$, and $z=z_1z_2z_3$ with $|x_2|\geq1$ and $|x_2x_3|\leq c$ such that $x_1x_2^ix_3y_1y_2^iy_3\in L'$ and $x_1x_2^ix_3z_1z_2^iz_3\in L'$ for any number $i\in\nat$.
Since $|x_2x_3|\leq c$, $x_2\in\{b_{j_1}\}^+$ follows.
Let $x_2=b_{j_1}^{e}$ with $e\geq1$ and set $i=3$.
Clearly, $x_1x_2^3x_3y_1y_2^3y_3$ has factors $a_{j_1}^{j_1n}$ and $b_{j_1}^{j_1n-2+3e}$, and thus we obtain $j_1n>j_1n+3e-2$, a contradiction.

(2)
We next discuss the non-membership of $NPal_d^{\#}$ to $\dcfl(d-1)$.
Let $A$ denote the regular language $\{w_1\# w_2\# \cdots \# w_{2d} \mid \forall i\in[2d](w_i\in\{0,1\}^*)\}$ and define $L' = A \cap (\{0,1,\#\}^* - NPal_d^{\#})$. Note that $L'$ equals $\{w_1\# \cdots \# w_d \# v_1\# \cdots \# v_d \mid \exists i\in[d](v_i=w_i^R)\}$. If $NPal_d^{\#}\in\dcfl(d-1)$, then Lemma \ref{hierarchy-properties}(2) implies that $L'\in\dcfl[d-1]$.
It therefore suffices to prove that $L'\notin\dcfl[d-1]$.

Assume to the contrary that $L'$ is in $\dcfl[d-1]$. We then apply to $L'$ the second pumping lemma for $\dcfl[d-1]$ (Lemma \ref{second-pumping-lemma}).
Let us choose a sufficiently large integer $n$ and take $d$ distinct strings $w_1,w_2,\ldots,w_d$ of length $n$. For each index $i\in[d]$, we set $v_i$ to be $w_i^Rs_i$, where $s_i$ is any string of length $n$.
We then define $x = w_1\#w_2\# \cdots \# w_d\#$ and $y^{(i)}= v_1\# v_2\# \cdots \# v_{i-1}\# w_i^R \# v_{i+1} \# \cdots \# v_d$ for any index $i\in[d]$.
Clearly, $xy^{(i)}$ is in $L'$ and $(w_i,w_i^R)$ is a nondegenerate iterative pair of $xy^{(i)}$ for $L'$.
The second pumping lemma provides two distinct indices $j_1,j_2\in[d]$ that satisfy one of the conditions (1)--(5) of the lemma.
Assuming $j_1<j_2$, we set $x'=w_1\# \cdots \# w_d\# v_1\#\cdots \# v_{j_1-1} \# w_{j_1}^R$, $y=\#v_{j_1+1}\# \cdots \# v_d$ and $z= s_{j_1}\#v_{j_1+1}\# \cdots \# v_{j_2-1}\# w_{j_2}^R \# v_{j_2+1}\# \cdots \# v_d$ so that $x'y = xy^{(j_1)}$ and $x'z = xy^{(j_2)}$ hold.

Obviously, $(w_{j_1},w_{j_1}^R)$ and $(w_{j_2},w_{j_2}^R)$ are respectively nondegenerate iterative pairs of $x'y$ and $x'z$ for $L'$.
Since there are no other iterative pairs of $x'y$ and $x'z$, the conditions (1)--(4) fail.
We then turn our attention to the condition (5). By assuming $u=y$, the condition (5) provides three factorizations $x'=x_1x_2x_3x_4x_5x_6x_7$, $y=y_1y_2y_3$, and $z=z_1z_2z_3$ such that (i) $(x_4x_5,x_7y_1)$ is a  nondegenerate iterative pair of $x'y$ for $L'$ and (ii) either $(x_2,z_2)$ or $(x_3x_4,z_2)$ is a nondegenerate iterative pair of $x'z$ for $L'$. It then follows that $x_4x_5=w_{j_1}$, $x_7=w_{j_1}^R$, and $y_1=\varepsilon$.
By the condition (5), it follows from (ii) that $z_2=w_{j_2}^R$ and either $x_2=w_{j_2}$ or $x_3x_4=w_{j_2}$. This contradicts the fact that $j_1<j_2$. The other case of $u=z$ is similarly handled.
\end{proof}


Recall that $MPal^{\#} =\{w_1\# w_2\# \cdots \# w_m \#^2  v_1\# v_2 \# \cdots \# v_n\mid m,n\geq1,
\exists i\in[\min\{m,n\}](|v_i|\neq |w_i| \vee v_i=w_i^R)\}$, which extends the language $L'$ defined in the proof of Theorem \ref{two-langauge-proof}(2), where $w_1,w_2,\ldots,w_m,v_1,v_2,\ldots,v_n \in\{0,1\}^*$.
Theorem \ref{Pal-DCFL} is rephrased in the following way.


\ms
\n{\bf Theorem \ref{Pal-DCFL}.} (rephrased)
\renewcommand{\labelitemi}{$\circ$}
\begin{enumerate}\vs{-2}
  \setlength{\topsep}{-2mm}%
  \setlength{\itemsep}{1mm}%
  \setlength{\parskip}{0cm}%
{\it
\item The language $Pal$ is in $\cfl\cap\co\cfl - \dcfl[\omega]$.

\item The language $MPal^{\#}$ is not in $\dcfl(\omega)\cup \dcfl[\omega]$. }
\end{enumerate}
As noted in Section \ref{sec:intersection-hierarchy}, the non-membership of $Pal$ to $\dcfl[\omega]$ was already proven in \cite[Example 2.1]{HH74} by applying their special iteration theorem for $\dcfl$ \cite[Theorem 2.3]{HH74}. Here, we wish to reprove the same claim by employing our own pumping lemma.

\ms
\begin{proof}
(1)
It is well-known that $Pal$ is in $\cfl$. Hence,
we wish to verify that $Pal$ is also in $\co\cfl$. Let $ODD=\{w\in\{0,1\}^*\mid |w| \text{ is odd }\}$.
Note that $\overline{Pal}$ coincides with $\{xy\mid |x|=|y|, x,y\in\{0,1\}^*,y\neq x^R\} \cup ODD$.
Consider the following 1npda. Given an input $w$,  choose either $0$ or $1$ nondeterministically. If $0$ is chosen, then we check whether $w\in ODD$. Otherwise, we split $w$ nondeterministically into $xy$, store $x$ into a stack, pop $x$ symbol by symbol while reading $y$, and check that $y=x^R$.
If $|x|\neq|y|$, then we reject $w$. Assume otherwise. If we discover any discrepancy between $y$ and $x^R$, then we accept $w$; otherwise, we reject it. It is obvious that this machine correctly recognizes  $\overline{Pal}$, and thus $Pal$ belongs to $\co\cfl$.

Next, we intend to verify that $Pal\notin \dcfl[d]$ for an arbitrary number  $d\geq2$ because this implies $Pal\notin\dcfl[\omega]$.
Toward a contradiction, we assume that $Pal\in\dcfl[d]$ for a certain index $d\geq2$. We apply  to $Pal$ the first pumping lemma for $\dcfl[d]$. Take a pumping-lemma constant $c>0$ guaranteed by the lemma and set $n=2c$.

Here, we call $0^n$ and $1^n$ a \emph{$0$-block} and a \emph{$1$-block}, respectively. For each index $k\in[d+1]$, we define $w_k = 0^n1^n0^n1^n \cdots 0^n$ ($4k+4d+3$ blocks). For example, $w_1=0^n1^n0^n$ and $w_2=0^n1^n0^n1^n0^n$.
There exists a pair $j_1,j_2\in[d+1]$ with $j_1<j_2$ satisfying the first pumping lemma.
Let $x'=0^n1^n0^n\cdots 0^n$ ($2j_2+2d+1$ blocks), $y=1^n0^n\cdots 0^n$ ($4j_1-2j_2+2d+2$ blocks), and $z=1^n0^n\cdots 0^n$ ($2j_2+2d+2$ blocks) so that $w_{j_1}=x'y$ and $w_{j_2}=x'z$.
Consider the condition (1) of the lemma. There exists a factorization $x'=x_1x_2x_3x_4x_5$ with $|x_2x_4|\geq1$ and $|x_2x_3x_4|\leq c$ such that $x_1x_2^ix_3x_4^ix_5y, x_1x_2^ix_3x_4^ix_5z\in Pal$ for any number $i\geq0$. Since $|x_2x_3x_4|\leq c$, $x_2x_3x_4$ must be a substring of $0^n$, $1^n$, $0^n1^n$, or $1^n0^n$.

Consider the case of $x_2x_3x_4\sqsubseteq 0^n$. Since $|x_2x_3x_4|\leq c$, $x_2x_3x_4$ must be included in the $(2j_1+2d+2)$th block. If we take $i=0$, then the size of this block is less than $n$. Since this block is situated in the left side of the center of $x'z$, $x_1x_3x_5z$ does not belong to $Pal$.
The other cases are similarly treated.

In the case where the condition (1) fails, let us consider three factorizations $x'=x_1x_2x_3$, $y=y_1y_2y_3$, and $z=z_1z_2z_3$ with $|x_2|\geq1$ and $|x_2x_3|\leq c$ satisfying $x_1x_2^ix_3y_1y_2^iy_3, x_1x_2^ix_3z_1z_2^iz_3\in Pal$ for any number $i\geq0$.  We conclude that $x_2$ is in the $(2j_2+2d+1)$th block of $x'$ because of $|x_2x_3|\leq c$. By taking $i=0$, the size of this block becomes less than $n$. Since this block is located in the right side of the center of $x'y$, we obtain $x_1x_3y_1y_3\notin Pal$, a contradiction.

(2)
We wish to show that $MPal^{\#}\notin \dcfl[\omega]$. Assume otherwise and choose a number $d\geq1$ satisfying that $MPal^{\#}\in \dcfl[d]$. We then restrict $MPal^{\#}$ and define $MPal^{\#}_{d+1}$ to be $\{w_1\# \cdots \# w_{d+1} \#  v_1\# \cdots \# v_{d+1} \mid
\exists i\in[d+1](|v_i|\neq |w_i| \vee v_i=w_i^R)\}$. An argument for the language $L'$ presented in the proof of Theorem \ref{two-langauge-proof}(2) can be modified appropriately to show that $MPal^{\#}_{d+1}\in \dcfl[d]$.

Finally, we want to prove that $MPal^{\#}\notin \dcfl(\omega)$. Toward a contradiction, we assume otherwise. For convenience, we define $FORM=\{x_1\# \cdots \# x_m \#^2 y_1\# y_2\# \cdots \# y_n\mid m,n\geq1,  x_1,\cdots, x_m,y_1,\ldots,y_n\in\{0,1\}^*\}$ and set  $NPal^{\#}_{=} = \{w_1\# \cdots \# w_m\#^2 v_1\# \cdots \# v_n\mid m,n\geq1,\forall i\in[\min\{m,n\}](|v_i|=|w_i|\wedge v_i\neq w_i^R)\}$, where  $x_1,\cdots x_m,y_1,\ldots,y_n\in\{0,1\}^*$.
Note that $NPal^{\#}_{=} = FORM\cap \overline{MPal^{\#}}$.
Since $FORM\in\reg$, by Lemma \ref{hierarchy-properties}(2), it follows from our assumption that $NPal^{\#}_{=} \in\dcfl[\omega]$.
We then take the minimum index $d\in\nat^{+}$ for which $NPal^{\#}_{=}$ belongs to $\dcfl[d]$.
As a special case of $NPal^{\#}_{=}$, we consider the language $L'' =
\{w_1\# \cdots \# w_{d+1}\#^2 v_1\# \cdots \# v_{d+1}\mid \forall i\in[d+1](|v_i|=|w_i|\wedge v_i\neq w_i^R)\}$.
Since $NPal^{\#}_{=}\in\dcfl[d]$ implies $L''\in\dcfl[d]$, $L''$ must be  written as the $d$-union $\bigcup_{e\in[d]} L_e$ for appropriate dcf languages $L_1,\ldots,L_d$.

Let $z$ denote any nonempty string of length $n$. Recall that $z[i]$ expresses the $i$th symbol of $z$ for any number $i\in[n]$. We introduce two notations. Let $z\pair{i}$ denote the string obtained from $z$ by flipping the $n$th bit of $z$. For instance, $0110\pair{1}=1110$, $0110\pair{2}=0010$, and $0110\pair{3}=0100$. Given two numbers $i,j\in[n]$ with $i< j$, we define $z\pair{i,j}$ by inductively setting  $z\pair{j,j}=z\pair{j}$ and $z\pair{i,j} = (z\pair{i+1,j})\pair{i}$. When $n$ is an even number, let $z^{(-)}_{\text{first}}$ (resp., $z^{(-)}_{\text{second}}$) denote the string obtained from $z$ by deleting the first half (resp., the second half) of symbols in $z$. Notice that $z= z^{(-)}_{\text{first}} z^{(-)}_{\text{second}}$.

From the second pumping lemma for $\dcfl[d]$ for $L''$, we take a pumping-lemma constant $c>0$. Choose a sufficiently large even number $n \geq 3c$ and  $d+1$ distinct strings $w_1,w_2,\ldots,w_{d+1}\in\{0,1\}^n$. For each index $i\in[d+1]$, we define $h_i$ to be the string of the form $w_1\# w_2\# \cdots \# w_{d+1} \#^2 v_1\# v_2\# \cdots \# v_{i-1} \# \hat{v}_i \# v_{i+1} \# \cdots \# v_{d+1}$, where $v_j=(w_j^R)\pair{n}$ and $\hat{v}_j = (w_j^R)\pair{n/2+1,n}$ for any $j\in[d+1]$.
These strings $h_1,h_2,\ldots,h_{d+1}$ obviously belong to $L''$. It follows that (*) for any index $i\in[d+1]$, any nondegenerate iterative pair $(u_1,u_2)$ of $h_{i}$ for $L''$ satisfies both $u_1\sqsubseteq w_i$ and $u_2\sqsubseteq v_i$ because of the requirement of $|v_i|=|w_i|$.

The second pumping lemma then ensures the existence of two indices $j_1,j_2\in[d+1]$ with $j_1<j_2$ such that $(h_{j_1},h_{j_2})$ satisfies one of the conditions (1)--(5) of the lemma.
We then define $x'=w_1\# \cdots w_{d+1}\#^2 v_1\# \cdots \# (v_{j_1})^{(-)}_{\text{first}}$, $y= (\hat{v}_{j_1})^{(-)}_{\text{second}} \# v_{j_1+1}\# \cdots \# v_{d+1}$, and $z= (v_{j_1})^{(-)}_{\text{second}} \# v_{j_1+1}\# \cdots \# v_{j_2-1}\# \hat{v}_{j_2} \# v_{j_2+1} \# \cdots \# v_{d+1}$. Note that $h_{j_1} = x'y$ and $h_{j_2}=x'z$. By the above statement (*), the conditions (1)--(4) do not hold. We thus look into the condition (5). A nondegenerate iterative pair $(x_4x_5,x_7u_1)$ of $x'y$ for $L''$ in the condition (5) must satisfy that $x_4x_5\sqsubseteq w_{j_1}$ and $x_7u_1\sqsubseteq \hat{v}_{j_1}$. Similarly, if a nondegenerate iterative pair $(x_3x_4,u'_2)$ of $x'z$ for $L''$ exists, then we obtain  $x_3x_4\sqsubseteq w_{j_2}$ and $u'_2\sqsubseteq \hat{v}_{j_2}$. However, these contradict the assumption of $j_1<j_2$. The case of $(x_2,u'_2)$ in the condition (5) is similarly treated.
Therefore, $L''$ is not in $\dcfl[d]$, as requested.
This completes the proof.
\end{proof}


Proposition \ref{extension-Hibbard} sharpens the hierarchy separation of Hibbard \cite{Hib67}, who proved $\klda{k}\neq \klda{(k+1)}$ for all  $k\in\nat^{+}$.

\ms
\n{\bf Proposition \ref{extension-Hibbard}.} (rephrased)
{\it For any $k\geq2$, $\klda{k}\cap \dcfl[2^{\floors{k/2}}] \nsubseteq \klda{(k-1)}\cup \dcfl[2^{\floors{k/2}}-1]$.}
\ms

\begin{proof}
In this proof, we intend to use a series of special languages, each of which is a slight modification of the language used in \cite[arXiv version]{Yam19}, introduced directly from Hibbard's \cite[Section 4]{Hib67} example language, and is shown to be in $\klda{k}$ but not in $\klda{(k-1)}$ for each $k\geq2$ \cite{Hib67,Yam19}. We call this new language by $L_k$.

Let $k\geq2$ be any integer. For each fixed index $i\in[k]$, we write $w_i$ for an arbitrary string of the form $a^{n_i}b^{m_i}c^{p_i}$ for any numbers $n_i,m_i,p_i\in\nat$.
The language  $L_k$ is composed of all strings $w$ of the form $w_2\#w_4\#\cdots \# w_{k-2} \# w_{k}\# w_{k-1} \# \cdots \#w_5\#w_3\# w_1$ if $k$ is even, and $w_2\#w_4\#\cdots \# w_{k-1} \# w_{k}\# w_{k-2} \# \cdots \#w_5\#w_3\# w_1$ if $k$ is odd, together with the following conditions (a)--(e).

\begin{enumerate}\vs{-2}
  \setlength{\topsep}{-2mm}%
  \setlength{\itemsep}{1mm}%
  \setlength{\parskip}{0cm}%

\item[(a)] For any $j\in[2,k]_{\integer}$, either $n_j\leq m_j\leq n_j+1$ or $m_j\leq p_j\leq m_j+1$ holds.

\item[(b)] It holds that $n_1\leq m_1 \leq n_1+1$.

\item[(c)] If $n_2\leq m_2\leq n_2+1$, then $n_1=m_1$ holds. If $m_2\leq p_2\leq m_2+1$, then $m_1=n_1+1$ holds.

\item[(d)] Let $j$ denote any index in $[3,k]_{\integer}$. If $n_j\leq m_j\leq n_j+1$, then either $n_{j-1}=m_{j-1}$ or $m_{j-1}=p_{j-1}$ holds. If $m_j\leq p_j\leq m_j+1$, then either $m_{j-1}= n_{j-1}+1$ or $p_{j-1} = m_{j-1}+1$ holds.

\item[(e)] Either $n_k=m_k$ or $m_k=p_k$ holds.
\end{enumerate}\vs{-2}

For instance, $L_2$ equals $\{a^{n_2}b^{m_2}c^{p_2}\# a^{n_1}b^{m_1}c^{p_1}\mid (n_1=m_1\wedge n_2=m_2) \text{ or } (m_1=n_1+1\wedge m_2=p_2)\}$. For convenience, we also define another language $L'_k$ obtained from $L_k$ by replacing the last term $w_1= a^{n_1}b^{m_1}c^{p_1}$ of $L_k$ with $a^{p_1}b^{m_1}c^{n_1}$ but we treat $(p_1,m_1,n_1)$ as $(n_1,m_1,p_1)$ in (a)--(e).

Following \cite[arXiv version]{Yam19}, $L_k$ can be proven to belong to $\klda{k}$ but it is not in $\klda{(k-1)}$. It thus suffices to show by induction on $k$ that $L_k\in\dcfl[2^{\floors{k/2}}]$ and $L_k\notin \dcfl[2^{\floors{k/2}}-1]$. To do so, we want to take the following two steps (1)--(2). For readability, we assume in (1)--(2) that $n_i$, $m_i$, and $p_i$
express numbers in $\nat$ for all $i$'s.

(1) As the first step, we prove the membership of $L_k$ to $\dcfl[2^{\floors{k/2}}]$.

(i) When $k=2$, we define $A_{(1)}=\{a^{n_2}b^{m_2}c^{p_2}\# a^{n_1}b^{m_1}c^{p_1} \mid n_1=m_1, n_2=m_2\}$ and $A_{(2)}=\{a^{n_2}b^{m_2}c^{p_2}\# a^{n_1}b^{m_1}c^{p_1} \mid m_1=n_1+1, m_2=p_2\}$. Clearly,  $L_2$ equals $A_{(1)}\cup A_{(2)}$, and thus $L_2$ belongs to $\dcfl[2]$. Similarly, we can prove that $L'_2$ is in $\dcfl[2]$.

(ii) In the case of $k=3$, we set $A_{(1)}=\{w \mid n_1=m_1, (n_2=m_2\wedge n_3=m_3) \text{ or } (m_2=n_2+1\wedge m_3=p_3)\}$ and $A_{(2)} =\{w\mid m_1=n_1+1, (m_2=p_2\wedge n_3=m_3) \text{ or } (p_2=m_2+1\wedge m_3=p_3)\}$. It follows that $L_3 = A_{(1)}\cup A_{(2)}\in \dcfl[2]$. A similar argument proves that $L'_3$ is in $\dcfl[2]$.

(iii) Consider the case where $k$ is an even number with $k\geq4$. Given a string $w$ of the form $w_2\# w_4 \# \cdots \# w_{k-2}\# w_k \# w_{k-1}\# \cdots \# w_3\# w_1$, we succinctly write $\tilde{w}^{(-)}$ for the string $w_4 \# \cdots \# w_{k-2}\# w_k \# w_{k-1}\# \cdots \# w_3$ so that $w$ equals $w_2 \# \tilde{w}^{(-)} \# w_1$.
Let $A_{(1)} = \{w\mid (n_1=m_1\wedge n_2=m_2 \wedge \tilde{w}^{(-)}\in L_{k-2}) \text{ or } (n_1=m_1 \wedge m_2=n_2+1\wedge \tilde{w}^{(-)}\in L'_{k-2}) \}$ and
$A_{(2)} = \{w\mid (m_1=n_1+1\wedge m_2=p_2 \wedge \tilde{w}^{(-)}\in L_{k-2}) \text{ or } (m_1= n_1+1\wedge p_2=m_2+1\wedge \tilde{w}^{(-)}\in L'_{k-2})\}$,  where $w$ is of the above-mentioned form.
By induction hypothesis, $L_{k-2},L'_{k-2}\in\dcfl[2^{\floors{(k-2)/2}}]$ implies that $A_{(j)}\in \dcfl[2^{\floors{(k-2)/2}}]$ for every index  $j\in[2]$. Since $L_k = A_{(1)}\cup A_{(2)}$, it instantly follows that $L_k\in\dcfl[2\cdot 2^{\floors{(k-2)/2}}]=\dcfl[2^{\floors{k/2}}]$.
A similar argument proves that $L'_k\in \dcfl[2^{\floors{k/2}}]$.

(iv) The case of odd $k\geq3$ is treated similarly to (iii).

(2)
As the second step, we prove that $L_k\notin \dcfl[2^{\floors{k/2}}-1]$. Assuming to the contrary that $L_k\in\dcfl[2^{\floors{k/2}}-1]$, we apply the pumping lemmas for $\dcfl[2^{\floors{k/2}}-1]$ to $L_k$.

As for the base case of $k=2$, the first pumping lemma for $\dcfl$ easily leads to $L_2\notin \dcfl$, from which we obtain $L_2\notin \dcfl[1]$.
Next, we focus on the case of $k\geq3$. Toward a contradiction, we assume that $L_k\in\dcfl[2^{\floors{k/2}}-1]$. In the case where $k$ is even, we first define the index set $I_k=\{1,2,3,4\}^{k-1}\times \{1,2\}$.
We then take $k$ sufficiently large (and thus larger than $c$) positive integers $n_1,n_2,\ldots,n_k$, which  are all different from each other. For each index $j\in[k]$, we set $w_j^{(1)} = a^{n_j}b^{n_j}c^{n_j-1}$, $w_j^{(2)}=a^{n_j}b^{n_j+1}c^{n_j}$, $w_j^{(3)} = a^{n_j}b^{n_j-1}c^{n_j-1}$, and $w_j^{(4)} = a^{n_j}b^{n_j-1}c^{n_j}$.
For a given string $s$ in $I_k$ of the form $s_2s_4\cdots s_{k-2}s_ks_{k-1}\cdots s_3s_1$, we
define $\alpha_s = w_2^{(s_2)} \# w_4^{(s_4)}\# \cdots \# w_{k-2}^{(s_{k-2})} \# w_{k}^{(s_{k})} \# w_{k-1}^{(s_{k-1})} \# \cdots \# w_3^{(s_3)}\# w_1^{(s_1)}$.
In what follows, we concentrate on such a series $s$ that makes $\alpha_s$ belong to $L_k$.
For the language $L_k$, let us recall the $2^{\floors{k/2}}$ dcf languages, say, $A_1,A_2,\ldots, A_{2^{\floors{k/2}}}$, which have been constructed recursively by (i)--(iv). For each $A_i$, we choose the smallest $s$ from $I_k$ such that $\alpha_s$ falls in $A_i$, where $s$ is viewed as an integer vector of dimension $k$. We then form a set of those minimum vectors $s$ and call it by $A_{vec}$. For instance, when $k=3$, $A_{vec}$ equals $\{111,312\}$; when $k=4$, $A_{vec}$ is $\{1111,1321,3112,3322\}$; when $k=5$, $A_{vec}$ matches $\{11111,13121,31112,33122\}$.

The second pumping lemma for $\dcfl[2^{\floors{k/2}}-1]$ provides a constant $c>0$ and a distinct pair $s,s'\in A_{vec}$.
Let $s=s_2s_4\cdots s_{k-2}s_ks_{k-1}\cdots s_3s_1$ and $s'=s'_2s'_4\cdots s'_{k-2}s'_ks'_{k-1}\cdots s'_3s'_1$.
We take the minimum even number $j$ satisfying $s_j\neq s'_j$.
Without loss of generality, we assume that $s_j<s'_j$.
Let us consider two corresponding strings $\alpha_s$ and $\alpha_{s'}$.
For any even number $j\in[k-2]$, we define $\alpha_s^{(j)} = w_2^{(s_2)} \# w_4^{(s_4)}\# \cdots \# w_{j}^{(s_j)}\#$ and $\beta_s^{(j)} =  w_{j+2}^{(s_{j+2})}\# \cdots \# w_3^{(s_3)} \#w_1^{(s_1)}$. When $j=k$, we instead set $\beta_s^{(j)} = w_{j-1}^{(s_{j-1})}\# \cdots \# w_3^{(s_3)}\#w_1^{(s_1)}$. Notice that $\alpha_s= \alpha_s^{(j)}\beta_s^{(j)}$ holds.

Our proof strategy is that, by choosing $\hat{x}$, $\hat{y}$, and $\hat{z}$ properly, we define    $x'=\alpha_s^{(j-2)}\#\hat{x}$, $y=\hat{y}\#\beta_s^{(j)}$, and $z=\hat{z}\#\beta_s^{(j)}$ to form $\alpha_s=x'y$ and $\alpha_{s'}=x'z$. We then apply the conditions (1)--(5) of the second pumping lemma to draw a contradiction.\footnote{For this purpose, we actually need to slightly modify the original lemma so that the conditions (3)--(4) always hold among all degenerate iterative pairs satisfying the conditions without any length requirement on $x_2$ and $x'_2$.}

There are four cases to consider separately.
(i$'$) If $(s_j,s'_j)=(1,3)$, then we set $\hat{x}=a^{n_j}b^{n_j-1}$, $\hat{y}=bc^{n_j-1}$, and $\hat{z}=c^{n_j-1}$.
(ii$'$) If $(s_j,s'_j)=(1,4)$, then we set $\hat{x}=a^{n_j}b^{n_j-1}$, $\hat{y}=bc^{n_j-1}$, and $\hat{z}=c^{n_j}$.
(iii$'$) If $(s_j,s'_j)=(2,3)$, then we set $\hat{x}=a^{n_j}b^{n_j-1}$, $\hat{y}=b^2c^{n_j}$, and $\hat{z}=c^{n_j-1}$.
(iv$'$) If $(s_j,s'_j)=(2,4)$, then we set $\hat{x}=a^{n_j}b^{n_j-1}$, $\hat{y}=b^2c^{n_j}$, and $\hat{z}=c^{n_j}$.

Finally, let us discuss the validity of the case (i$'$). Consider any  degenerate iterative pair $(x'_2,z_2)$ of $x'z$ for $L_k$, where $x'_2$ is in $\hat{x}$ with $|x'_2|\leq c$ and $z_2$ is in $\hat{z}$.
However, there is no degenerate iterative pair for $x'y$ in $\hat{x}\hat{y}$, a contradiction with the condition (4).
For the other cases, we can use a similar argument to draw a contradiction. 

This completes the proof of the proposition.
\end{proof}


\section{Proof of the First Pumping Lemma for DCFL[$d$]}\label{sec:proof-pumping}

Throughout this section, we will provide the proof of the first pumping lemma for $\dcfl[d]$ (Lemma \ref{pumping-lemma}).
As a caviar of this exposition, we will employ  a novel machine-based argument to prove Lemma \ref{pumping-lemma-state}, which leads to the basis case of  $d=1$ in Lemma \ref{pumping-lemma}.
This also provides an alternative proof to Yu's pumping lemma.


In Section \ref{sec:state-stack}--\ref{sec:mutual-correlation}, we will explain fundamental notions needed for the proofs.
These notions loosely come from \cite{Yam08,Yam16}. The  proof of Lemma \ref{pumping-lemma-state} will be given in Section \ref{sec:lemma-verify} and the proof of Lemma \ref{pumping-lemma} will be provided in Section \ref{sec:first-lemma-proof}.

\subsection{Boundaries, Turns, and $\varepsilon$-Enhanced Strings}\label{sec:state-stack}

Let us begin with two elementary notions of \emph{boundaries} and \emph{boundary blocks}, which are necessary to introduce other important notions.
For this purpose, we visualize a \emph{single move} of a 1dpda $M$ as a series of three consecutive actions: (i) firstly, replacing a topmost stack symbol, (ii) secondly, updating an inner state, and (iii) thirdly, either moving a tape head to the right or making it stay still (in the case of an  $\varepsilon$-move).

A \emph{boundary} is a borderline between two adjacent tape cells (except for the case of the leftmost tape cell). We index all such boundaries from $0$ to $|\cent x\dollar|$ as follows. The boundary $0$ is located at the left borderline of cell $0$ and boundary $i+1$ is located  between  cell $i$ and cell $i+1$ for each number $i\in\nat$.
We often use both ``cell indices'' and ``boundaries'' to specify an area of the input tape interchangeably.
When a string $xy$ is written in $|xy|$ consecutive cells, the \emph{$(x,y)$-boundary} refers to the boundary $|x|+1$, which separates between $x$ and $y$.
A \emph{boundary block} between two boundaries $t_1$ and $t_2$ with $t_1<t_2$ is a consecutive series of boundaries between $t_1$ and $t_2$ (including $t_1$ and $t_2$). These numbers $t_1$ and $t_2$ are called the \emph{ends} of this boundary block. For brevity, we write $[t_1,t_2]$ to denote such a boundary block between $t_1$ and $t_2$.
For each boundary block $[t_1,t_2]$, its \emph{fringe} is either the boundary $t_1-1$ or the boundary $t_2+1$ if they actually exist.
The \emph{$(t_1,t_2)$-region} indicates all the consecutive cells located in the boundary block $[t_1,t_2]$.  When an input string $x$ is written in this $(t_1,t_2)$-region, we conveniently call the region the \emph{$x$-region} as long as $t_1$ and $t_2$ are clear from
the context. On the contrary, $x$ is called the \emph{$(t_1,t_2)$-regional string}.


\begin{figure}[t]
\centering
\includegraphics*[height=4.8cm]{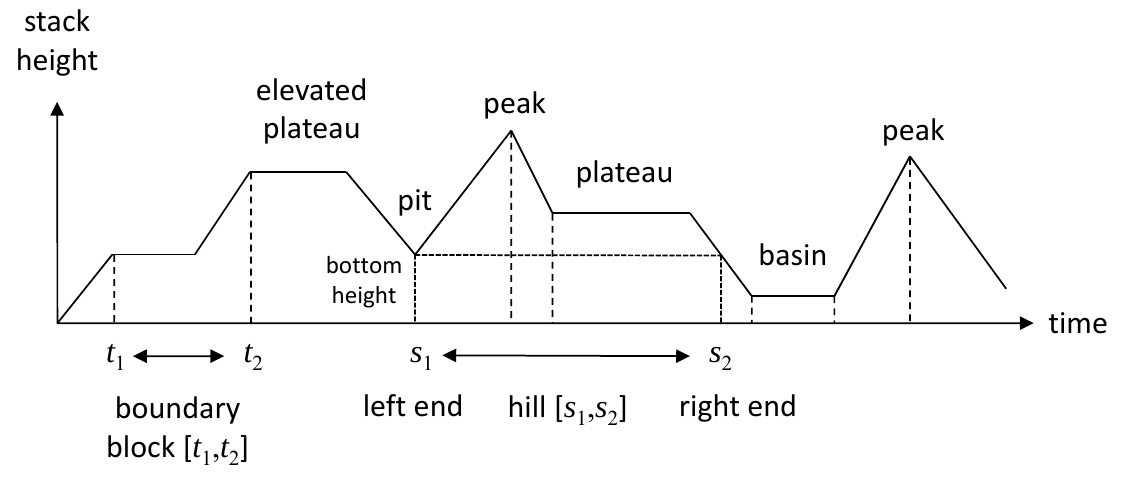}
\caption{A history of the changes of stack heights.}\label{fig:stack-height}
\end{figure}


The \emph{stack height of $M$ at boundary $t$} is the length of the stack content while the tap head is passing though this boundary $t$. Recall that the stack operation must be performed before the tape head move.
Given $t_1,t_2\in\nat$ with $t_1<t_2$,
the boundary block $[t_1,t_2]$ is \emph{flat} if the stack height does not change in the $(t_1,t_2)$-region. A boundary block $[t_1,t_2]$ is called \emph{convex} if there is a boundary $s$ between $t_1$ and $t_2$ (namely, $s\in[t_1,t_2]$) such that there is no pop operation in the $(t_1,s)$-region and there is no push operation in the $(s,t_2)$-region.
A boundary block $[t_1,t_2]$ is \emph{pseudo-convex} if the stack height at any boundary $s\in[t_1,t_2]$ does not go below $h_1 + \frac{h_2-h_1}{t_2-t_1}(s-t_1)$, where $h_i$ is the stack height at the boundary $t_i$ for any index $i\in\{1,2\}$. By their definitions, either convex or flat boundary blocks are pseudo-convex.

A \emph{peak} (resp., a \emph{pit}) is a boundary $t$ for which the stack heights at the boundaries $t-1$ and $t+1$ are smaller (resp., greater) than the stack height at the boundary $t$. A \emph{plateau} is a flat boundary block and the stack height of this boundary block is called the \emph{height} of the plateau.
A \emph{basin} (resp., an \emph{elevated plateau}) $[t,t']$ is a plateau  for which the stack heights at all fringes of this plateau are greater (resp., smaller) than the height of the plateau whenever the fringes exist.
A \emph{hill} is a boundary block $[t,t']$ such that (i) the stack height at the boundary $t$ and the stack height at the boundary $t'$ coincide, (ii) there exists either a peak or an elevated plateau, but not both, in $[t,t']$, (iii) there is no pit or basin in $[t,t']$, and (iv) one of the ends of this block is either a pit or the left edge of a basin. The lowest stack height of the hill is particularly called the \emph{bottom height}.
The \emph{height of the hill} is the difference between its bottom height and its highest stack height.
We briefly call either a peak or an elevated plateau existing in a hill by a \emph{hill top}.
Those concepts are illustrated in
Fig.~\ref{fig:stack-height}.

An \emph{upward slope} means a non-flat boundary block in which the stack height never decreases. In contrast, a \emph{downward slope} is a non-flat boundary block in which no increase of the stack height occurs. A \emph{slope} refers to either an upward slope or a downward slope.

A \emph{turn} is intuitively a change of stack height from ``nondecreasing'' to ``decreasing'' discussed in 1966 by Ginsburg and Spanier  \cite{GS66} (who actually defined ``turn'' in a slightly different way). Here, we
describe this notion in terms of boundary blocks and use it to partition an entire stack history.
A \emph{turning point} is either a peak or the right edge of an elevated plateau.
A \emph{turn} refers to a boundary block $[t_1,t_2]$ in which there is exactly one turning point  such that the left fringe (if any) is a part of a downward slope and the right fringe is a part of either an upward slope or a basin.
The \emph{height of a turn} is the difference between the maximum stack height and the minimum stack height in the turn. The minimum stack height is observed at either end of the turn and is also called the \emph{bottom height} of the turn.

Given strings over an alphabet $\Sigma$, \emph{$\varepsilon$-enhanced strings} are strings over the extended alphabet $\Sigma_{\varepsilon}$ ($= \Sigma\cup\{\varepsilon\}$). Here, the notation $\varepsilon$ is treated as a special input symbol expressing the ``absence'' of symbols, not the usual ``empty string.''
For such an $\varepsilon$-enhanced string $x$, we keep the notation $|x|$ to denote the total number of non-$\varepsilon$-symbols in $x$; on the contrary, we use the new notation $|x|_{\varepsilon}$ to denote the length of $x$ by counting $\varepsilon$ as an independent  symbol. For instance, if $x=001\varepsilon 0110\varepsilon10$, then we obtain $|x|=9$ and $|x|_{\varepsilon}=11$.
An \emph{$\varepsilon$-enhanced 1dpda} (or an $\varepsilon$-1dpa, for short) is a 1dpda that takes $\varepsilon$-enhanced input strings and works as a standard 1dpda except that a tape head always moves to the right without stopping by treating $\varepsilon$ as an actual symbol. This tape head movement is sometimes called ``real time.''

\begin{lemma}\label{enhanced-lemma}
For any $n$-state 1dpda $M$ with stack alphabet size $m$ and push size $e$, there exists an $\varepsilon$-1dpda $N$ with the same $n,m,e$ such that, for any input string $x$, there is an appropriate $\varepsilon$-enhanced string $\hat{x}$ for which $M$ accepts (resp., rejects) $x$ iff $N$ accepts (resp., rejects) $\hat{x}$, where $x$ is obtainable from $\hat{x}$ by deleting all $\varepsilon$ in $\hat{x}$. Moreover,
$\hat{x}$ is uniquely determined from $x$ and $M$. In this case, $\hat{x}$ is said to be \emph{induced} from $x$ by $M$. In addition, if $M$ is in an ideal shape, then $N$ can be made to be in an ideal shape as well.
\end{lemma}

\begin{proof}
Let $M$ denote any $n$-state 1dpda with stack alphabet size $m$ and push size $e$. The desired machine $N$ is defined as follows. On any $\varepsilon$-enhanced input string, whenever $M$ makes a non-$\varepsilon$-move, $N$ simulates this single step.
To simulate $M$'s $\varepsilon$-move, if $N$ can read a symbol $\varepsilon$ on an input tape, then $N$ simulates $M$'s move and then moves its tape head to the right; otherwise, $N$ instantly enters a rejecting state. This forces the input string to match $M$'s $\varepsilon$-moves and thus ensures its uniqueness. We then define an $\varepsilon$-enhanced string so that $N$ reads $\varepsilon$ exactly when $M$ makes an $\varepsilon$-move.
If $M$ is in an ideal shape, then the above transformation of $M$ together with an associated $\varepsilon$-enhanced string makes $N$ be in an ideal shape.
\end{proof}


\begin{figure}[t]
\centering
\includegraphics*[height=4.3cm]{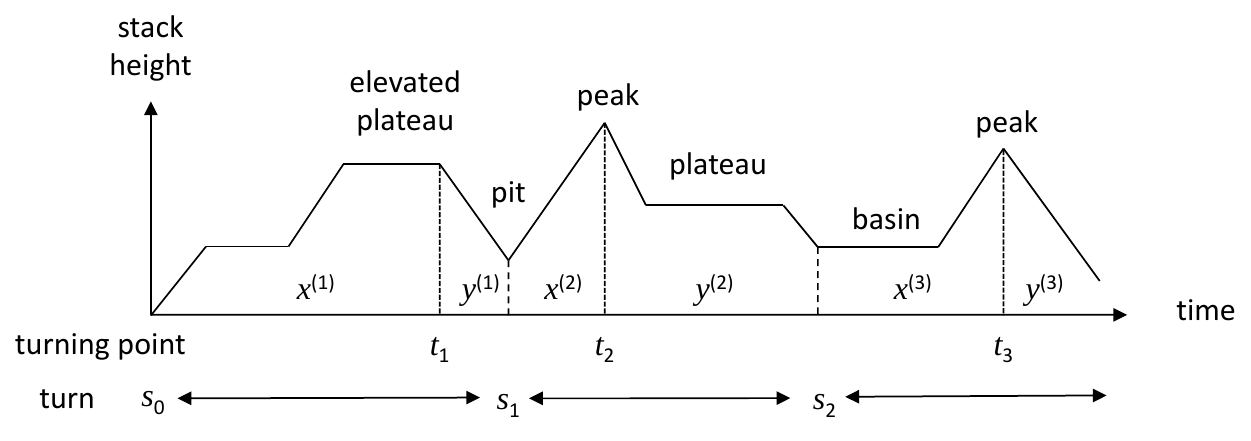}
\caption{A turn partitioning of a stack history. Each $x^{(i)}$ is the left area of a turn and $y^{(i)}$ is the right area of the same  turn.}\label{fig:turn-partition}
\end{figure}


The ideal shape condition of a 1dpda or an $\varepsilon$-1dpda makes it  possible to partition an entire stack history into consecutive turns, as shown in Fig.~\ref{fig:turn-partition}.  We call such a partition a \emph{turn partition}.

\subsection{State-Stack Pairs and Mutual Correlations}\label{sec:mutual-correlation}

We next introduce another important notion of state-stack pairs.
Let $M$ denote either a 1dpda or an $\varepsilon$-1dpda with a set $Q$ of inner states and a stack alphabet $\Gamma$, and assume that $M$ is in an ideal shape. A \emph{state-stack pair} at boundary $i$ is a pair $(q,\gamma)$ of inner state $q$ and stack content $\gamma$ (which is expressed as $a_1a_2\cdots a_k$ from the top to the bottom of the stack with $a_k=Z_0$),
namely, $(q,\gamma)$ in $Q\times (\Gamma^{(-)})^*Z_0$.
In particular, $(q_0,Z_0)$ is the special state-stack pair at the boundary $0$.
In a computation of $M$ on input $x$, a state-stack pair $(q,\gamma)$ at boundary $i$ with $i\geq1$ refers to the machine's current status where
$M$ is reading a certain input symbol, say, $\sigma$ at cell $i-1$ in a certain inner state, say, $p$ with a certain stack content $a\gamma'$ for $a\in \Gamma$, and $M$ then changes $p$ to $q$, modifying $a$ by either pushing another symbol $b$ satisfying $\gamma=ba\gamma'$ or popping $a$ to make $\gamma=\gamma'$ (due to the ideal shape condition).
It is important to remark that any computation of $M$ on $x$ can be expressed as a series of state-stack pairs at every boundary in the $\cent x\dollar$-region.
Consider a sequence $(\alpha_0,\alpha_1,\alpha_2,\ldots,\alpha_m)$ in which each $\alpha_i$ is a state-stack pair at step $i\in[0,m]_{\integer}$. Such a sequence completely describes a computation of $M$ if $\alpha_0=(q_0,Z_0)$ and $M$ halts in exactly $m$ steps.

Two boundaries $t_1$ and $t_2$ with $t_1<t_2$ are \emph{mutually correlated} if (1) there exist two state-stack pairs $(q,\gamma)$ and $(p,\gamma)$ at the boundaries $t_1$ and $t_2$, respectively, and (2) the boundary block $[t_1,t_2]$ is pseudo-convex. For four boundaries $t_1,t_2,t_3,t_4$ with $t_1<t_2<t_3<t_4$, the two boundary blocks $[t_1,t_2]$ and $[t_3,t_4]$ are \emph{mutually correlated} if (1$'$) $[t_1,t_2]$, $[t_2,t_3]$, and $[t_3,t_4]$ are all pseudo-convex, (2$'$) there exist two state-stack pairs $(q,\gamma)$ and $(p,\alpha\gamma)$ at the boundaries $t_1$ and $t_2$, respectively, and (3$'$) there exist two state-stack pairs $(s,\alpha\gamma)$ and $(r,\gamma)$ at the boundaries $t_3$ and $t_4$, respectively, for properly chosen $p,q,r,s\in Q$, $\gamma\in(\Gamma^{(-)})^*Z_0$, and $\alpha\in(\Gamma^{(-)})^*$.
Mutually correlated boundaries or boundary blocks are closely related to the iteration of certain parts of strings.
Furthermore, a pair $([t_1,t_2],[t_3,t_4])$ of boundary blocks with $t_1<t_2<t_3<t_4$ is called \emph{good} if the two boundary blocks $[t_1,t_2]$ and $[t_3,t_4]$ are mutually correlated, inner states at the boundaries $t_1$ and $t_2$ coincide, and inner states at the boundaries $t_3$ and $t_4$ coincide.


\begin{lemma}\label{crossing-property}
Let $M$ denote any $\varepsilon$-1dpda in an ideal shape and let $w$ be any $\varepsilon$-enhanced string with $|w|\geq1$.
\renewcommand{\labelitemi}{$\circ$}
\begin{enumerate}\vs{-2}
  \setlength{\topsep}{-2mm}%
  \setlength{\itemsep}{1mm}%
  \setlength{\parskip}{0cm}%

\item Let $t_1,t_2\in\nat$ with $1\leq t_1<t_2\leq |w|+1$. Let $w=x_1x_2x_3$ with $|x_2|\geq1$ be a factorization of string $w$ such that $t_1$ is the $(x_1,x_2)$-boundary and $t_2$ is the $(x_2,x_3)$-boundary. If the boundaries $t_1$ and $t_2$ are mutually correlated and two inner states at these boundaries coincide, then it follows that, for any number $i\in\nat$, $x_1x_2x_3\in L(M)$ iff $x_1x_2^ix_3\in L(M)$.

\item Let $t_1,t_2,t_3,t_4\in\nat$ with $1\leq t_1<t_2<t_3<t_4\leq |w|+1$. Let $w=x_1x_2x_3x_4x_5$ so that each $t_i$ is the    $(x_i,x_{i+1})$-boundary for each index $i\in [4]$. If the pair $([t_1,t_2],[t_3,t_4])$ is good, then $(x_2,x_4)$ is an iterative pair of $w$ for $L(M)$.
\end{enumerate}
\end{lemma}

\begin{proof}
(1) Assume that $1\leq t_1<t_2\leq |w|+1$. Since the boundaries $t_1$ and $t_2$ are mutually correlated, the boundary block $[t_1,t_2]$ is pseudo-convex and there exist two state-stack pairs of the form  $(q,\gamma)$ and $(p,\gamma)$ at the boundaries $t_1$ and $t_2$, respectively.
By the premise of the lemma, we also obtain $p=q$. Since any stack height in $[t_1,t_2]$ does not go below $|\gamma|$, it is possible to repeat this boundary block $[t_1,t_2]$ for an arbitrary number of times without affecting the rest of $M$'s computation; in particular, the outcomes of the computation does not alter. It therefore follows that, for any number $i\in\nat$,  $x_1x_2x_3\in L(M)$ iff $x_1x_2^i x_3\in L(M)$.

(2) Given any $t_1,t_2,t_3,t_4\in\nat$ with $1\leq t_1<t_2<t_3<t_4\leq |w|+1$, since the boundary blocks $[t_1,t_2]$ and $[t_3,t_4]$ are mutually  correlated, the boundary blocks $[t_1,t_2]$, $[t_2,t_3]$, and $[t_3,t_4]$ are all pseudo-convex and there are state-stack pairs $(q,\gamma)$ and $(p,\alpha\gamma)$ at the boundaries $t_1$ and $t_2$ and also $(s,\alpha\gamma)$ and $(r,\gamma)$ at the boundaries $t_3$ and $t_4$ for appropriate $p,q,s,r,\alpha,\gamma$.
The premise of the lemma then requires that $p=q$ and $r=s$. The pseudo-convexity of $[t_1,t_2]$, $[t_2,t_3]$, and $[t_3,t_4]$ makes it possible to repeat $[t_1,t_2]$ and $[t_3,t_4]$ for the same number of times without tampering the outcomes of $M$. Thus, $(x_2,x_4)$ is an iterative pair of $w$ for $L(M)$.
\end{proof}


To support the proof of Lemma \ref{pumping-lemma-state} in Section \ref{sec:lemma-verify}, we present another useful lemma. In this lemma, for readability, we abbreviate a concatenation $\alpha_{i}\alpha_{i-1} \cdots \alpha_1$ of $i$ strings as $\alpha_{[i,1]}$. In particular, $\alpha_{[1,1]}$ coincides with $\alpha_1$.

\begin{lemma}\label{mutual-matching}
Let $M$ be any $\varepsilon$-1dpda with $Q$ and $\Gamma$ and let $m$ denote  any integer with $m\geq2$. We fix an $\varepsilon$-enhanced input string $w$ and a computation of $M$ on $w$.  Let $p,q,p_i,q_i\in Q$ for all $i\in[m]$, $\gamma\in (\Gamma^{(-)})^*Z_0$, and $\alpha_i\in(\Gamma^{(-)})^*$ for any index $i\in[m-1]$.
Assume that there are boundaries $t_1,t_2,\ldots,t_m$ with $1\leq t_1<t_2<\cdots <t_m$ for which $(q,\gamma),(q,\alpha_1\gamma), \ldots, (q,\alpha_{[m-2,1]}), (q,\alpha_{[m-1,1]}\gamma)$ are state-stack pairs at the boundaries $t_1,t_2,\ldots,t_m$, respectively, and that there are boundaries $r_1,r_2,\ldots,r_m$ with $t_m\leq r_1<r_2<\cdots <r_m$ for which $(p_1,\alpha_{[m-1,1]}\gamma),(p_2,\alpha_{[m-2,1]}\gamma),\ldots, (p_{m-1},\alpha_1\gamma), (p_m,\gamma)$ are state-stack pairs at the boundaries $r_1,r_2,\ldots,r_m$, respectively. Moreover, each pair of boundary blocks $[t_i,t_{i+1}]$ and $[r_i,r_{i+1}]$ are pseudo-convex for each index $i\in[m-1]$.

\renewcommand{\labelitemi}{$\circ$}
\begin{enumerate}\vs{-2}
  \setlength{\topsep}{-2mm}%
  \setlength{\itemsep}{1mm}%
  \setlength{\parskip}{0cm}%

\item If $m>|Q|$, then there exist two distinct indices $i,j\in[m]$ satisfying $p_i=p_j$.

\item If $m>|Q|^2$, then there exists a subset $C$ of $[m]$ of size more than $|Q|$ such that $p_i=p_j$ holds for any pair $i,j\in C$.
\end{enumerate}
\end{lemma}

\begin{proof}
(1) Consider the set $P=\{p_1,p_2,\ldots,p_m\}$ of inner states given in the premise of the lemma. Since $P\subseteq Q$, it follows that $|P|\leq|Q|<m$. Thus, there is a distinct pair $i,j\in[m]$ satisfying $p_i=p_j$ by the pigeonhole principle.

(2) Assume that there is no such subset $C$ stated in the lemma. For each inner state $q\in Q$, we introduce $C_q=\{i\in[m]\mid p_i=q\}$.
Note that $\{C_q\}_{q\in Q}$ forms a partition of $[m]$. Since $|C_q|\leq |Q|$ for all $q\in Q$, it follows that $m=|\bigcup_{q\in Q}C_q| = \sum_{q\in Q}|C_q| \leq |Q|^2$. This equality is in contradiction with $m>|Q|^2$.
\end{proof}

\subsection{Proof of Lemma \ref{pumping-lemma-state}}\label{sec:lemma-verify}

With the useful lemmas shown in Section \ref{sec:mutual-correlation}, we hereafter describe the proof of  Lemma \ref{pumping-lemma-state} by analyzing the behaviors of 1dpda's.
Now, we fix an arbitrary 1dpda $M=(Q,\Sigma, \{\cent,\dollar\},\Gamma, \delta,q_0,Z_0,Q_{acc},Q_{rej})$ in an ideal shape recognizing $L(M)$ with $|Q|=n$ and expand it to an ``equivalent'' $\varepsilon$-1dpda by Lemma \ref{enhanced-lemma}. For the desired constant $c$, we set $c= 2^{6|Q|^6}$.
Firstly, we take two arbitrary strings $w_1=xy$ and $w_2=xz$ in $L$ over $\Sigma$ with $|y|,|z|\geq1$ and $|x|>c$ such that $w_1$ and $w_2$ have iterative pairs for $L(M)$.
To use the $\varepsilon$-enhanced machine, we also expand these input strings and introduce their corresponding $\varepsilon$-enhanced strings. For readability, we continue using the same notations $w_1$, $w_1$, $x$, $y$, $z$, and $M$ for the corresponding $\varepsilon$-enhanced strings and the corresponding $\varepsilon$-1dpda.

We assume that the condition (1) of the lemma fails. Our goal is then to verify that the condition (2) of the lemma indeed holds.
There are four specific cases to deal with. Hereafter, we intend to discuss them separately. Remember that, since $M$'s tape head moves in only one direction, any state-stack pair at each boundary in the $x$-region does not depend on the choice of $y$ and $z$.

{\bf Case 1:}
Let us consider the case where there are two boundaries $t_1,t_2$ with $1\leq t_1<t_2\leq |x|$ and $|t_2-t_1|\leq c$ such that (i) the boundaries $t_1$ and $t_2$ are mutually correlated and (ii) the inner states at the boundaries $t_1$ and $t_2$ coincide.
In this case, we can factorize $x$ into $x_1x_2x_3$ so that $t_1$ is the $(x_1,x_2x_3)$-boundary and $t_2$ is the $(x_1x_2,x_3)$-boundary. By Lemma \ref{crossing-property}(1), it then follows that
$x_1x_2^ix_3 y\in L$ and $x_1x_2^ix_3 z\in L(M)$ for any number $i\in\nat$. Thus, the condition (1) holds. This is a contradiction with our assumption.

{\bf Case 2:}
Consider the case where there are four boundaries $t_1,t_2,t_3,t_4$ with $1\leq t_1<t_2<t_3<t_4\leq |x|$ and $|t_4-t_1|\leq c$ and there are four elements $p,q\in Q$, $\gamma\in(\Gamma^{(-)})^*Z_0$, and $\alpha\in(\Gamma^{(-)})^*$ such that (i) $(q,\gamma)$ and  $(q,\alpha\gamma)$ are state-stack pairs at the boundaries $t_1$ and $t_2$, respectively, (ii)  $(p,\alpha\gamma)$ and $(p,\gamma)$ are two state-stack pairs at the boundaries $t_3$ and $t_4$, respectively, and
(iii) the boundary block $[t_i,t_{i+1}]$ for each index $i\in[3]$ is pseudo-convex.
We then factorize $x$ into $x_1x_2x_3x_4x_5$ so that $t_i$ is the $(x_i,x_{i+1})$-boundary for each index $i\in[4]$. Note that $|x_2x_4|\geq2$ because of $t_1<t_2$ and $t_3<t_4$.
Since $M$ is in an ideal shape, $|x_2|\geq1$ also follows. By an application of Lemma \ref{crossing-property}(2), we conclude that
$(x_2,x_4)$ is an iterative pair of both $xy$ and $xz$ for $L(M)$.
This also contradicts our assumption.

{\bf Case 3:}
For convenience, we set $R=(|x|-c,|x|)$. We first consider the simple case where there is no pop operation in the $R$-region. We can take a number $m\geq1$ and a series of $m$  boundaries $s_1,s_2,\ldots,s_m$ in the $R$-region with $1\leq s_1<s_2<\cdots <s_m$ such that, for appropriately chosen elements $q\in Q$,  $\gamma\in(\Gamma^{(-)})^*Z_0$, and $\alpha_1,\ldots,\alpha_{m-1}\in(\Gamma^{(-)})^*$, there exist state-stack pairs of the form $(q,\gamma),(q,\alpha_1\gamma),\ldots,(q,\alpha_{[m-1,1]}\gamma)$ respectively at the boundaries $s_1,s_2\ldots,s_m$,  where $\alpha_{[i,1]}$ is the abbreviation of $\alpha_i\alpha_{i-1}\cdots \alpha_1$ introduced in Section \ref{sec:mutual-correlation}.
Since no pop operation is performed, the boundary blocks $[s_1,s_2],[s_2,s_3],\ldots,[s_{m-1},s_m]$ are all convex. By maximizing the value of $m$, we hereafter assume that $m$ takes this maximum value. Notice that $R$-region contains more than $|Q|^3$ boundaries because of $|Q|^3<c$.
If, for each $q\in Q$, the corresponding maximum value $m_q$ is at most $|Q|^2$, then the size $c$ of the $R$-region must be at most $|Q|^3$. This is a contradiction, and thus we conclude that $m>|Q|^2$.

Next, we choose two series $\{t_i\}_{i\in[m]}$ and $\{r_i\}_{i\in[m]}$ of boundaries in the $y$-region and the $z$-region, respectively, for which
(a) $s_m\leq t_1<t_2<\cdots <t_m$ and $s_m\leq r_1<r_2<\cdots <r_m$ and
(b) for any index $i\in[m-1]$, $[s_i,s_{i+1}]$ is mutually correlated to $[t_{m-i},t_{m-i+1}]$ in the $y$-region and also to $[r_{m-i},r_{m-i+1}]$ in the $z$-region.
This is possible because $M$ is $\varepsilon$-enhanced and in an ideal shape.
Notice that the boundary blocks $[t_1,t_2],\ldots,[t_{m-1},t_m]$ and  $[r_1,r_2],\ldots,[r_{m-1},r_m]$ are all pseudo-convex.
Assume that, at the boundaries $t_1,t_2,\ldots,t_m$, the associated state-stack pairs are respectively of the form $(p_1,\alpha_{[m-1,1]}\gamma), (p_2,\alpha_{[m-2,1]}\gamma), \ldots,(p_{m-1},\alpha_1\gamma), (p_{m},\gamma)$ for appropriate inner states $p_1,p_2,\ldots,p_m\in Q$.
Similarly, assume that $(e_1,\alpha_{[m-1,1]}\gamma),(e_2,\alpha_{[m-2,1]}\gamma), \ldots,(e_{m-1},\alpha_1\gamma), (e_m,\gamma)$ are respectively state-stack pairs at the boundaries $r_1,r_2,\ldots,r_m$.

By Lemma \ref{mutual-matching}(2), since $m>|Q|^2$, there is a subset $C$ of $[m]$ of size more than $|Q|$ such that, for any index pair $i,j\in C$, the inner states at the boundaries  $t_i$ and $t_j$ coincide; namely, $p_i=p_j$. From this fact, it follows that there is a special pair $j_1,j_2\in C$ with $j_1<j_2$ for which the inner states at the boundaries $r_{j_1}$ and $r_{j_2}$ coincide; that is, $e_{j_1}=e_{j_2}$. We fix such a pair $(j_1,j_2)$.
We then factorize the strings $x$, $y$, and $z$ as $x=x_1x_2x_3$, $y=y_1y_2y_3$, and $z=z_1z_2z_3$ so that $s_{j_1}$ is the  $(x_1,x_2)$-boundary, $s_{j_2}$ is the $(x_2,x_3)$-boundary, $t_{j_1}$ is the $(y_1,y_2)$-boundary, $t_{j_2}$ is the $(y_2,y_3)$-boundary, $r_{j_1}$ is the $(z_2,z_3)$-boundary, and $r_{j_2}$ is the  $(z_2,z_3)$-boundary.
It then follows that
$(x_2,y_2)$ and $(x_2,z_2)$ are respectively iterative pairs of $xy$ and $xz$ for $L$. Therefore, the condition (2) follows.

{\bf Case 4:}
Assume that Cases 1--3 all fail. Hence, at least one pop operation must take place in the $R$-region. We will proceed our proof by following the steps (1)--(5) described below.
We hereafter focus our attention only on the $R$-region.

(1) We first claim that any flat boundary block has size at most $|Q|$.
This is because if a flat boundary block has size more than $|Q|$, then Case 1 occurs, a contradiction with  our assumption. Therefore, the claim is true. This implies that there are at least $c/(|Q|+1)$ operations of ``pop'' and ``push'' in the $R$-region.

(2)
Choose two boundaries $s$ and $s'$ satisfying the following three conditions:  $|x|-c\leq s < s' \leq |x|$, the boundary block $[s,s']$ is pseudo-convex, and $[s,s']$ can be  partitioned into a number of turns.
Such a turn partitioning is illustrated in Fig.~\ref{fig:turn-partition}. Hereafter, we discuss the specific  cases (a)--(b).

(a)
We first consider the case where $[s,s']$ is
a single turn. We focus on a hill, say, $[t_1,t_2]$ that is contained in this  turn and claim that this hill has height at most $|Q|^2$. This claim is shown as follows. If the hill has height more than $|Q|^2$, then
there exists an inner state $q\in Q$ that appears in at least $|Q|+1$ state-stack pairs in the left area of the hill. Let $(q,\gamma),(q,\alpha_1\gamma),\ldots,(q,\alpha_{[|Q|,1]}\gamma)$ denote those state-stack pairs for appropriately chosen strings $\gamma,\alpha_1,\ldots,\alpha_{|Q|+1}$. Associated with them, we also choose state-stack pairs $(p_1,\alpha_{[|Q|,1]}\gamma),\ldots, (p_{|Q|},\alpha_1\gamma),(p_{|Q|+1},\gamma)$ in the right area of the hill for certain inner states $p_1,\ldots,p_{|Q|+1}$.
Lemma \ref{crossing-property}(1) then guarantees the existence of a distinct pair $i,j\in[|Q|+1]$ satisfying $p_i=p_j$.
This leads to Case 2, a contradiction.

As a consequence, the total size of the boundary blocks, each of which  expresses one slope (either upward or downward) of the hill (excluding all plateaus), is at most $2|Q|^2$. Notice that each plateau has size at most $|Q|$ by (1).
No distinct pair of plateaus in the hill has the same height since, otherwise, those two plateaus contain two state-stack pairs with the same inner state, and thus Case 1 occurs, a contradiction. Thus, the hill includes no more than $|Q|^2$ plateaus. From this fact, it follows that  $|t_2-t_1|\leq 2|Q|^2+|Q|^2\times |Q| \leq 3|Q|^3$.
Note that the size $|s'-s|$ of the turn equals the size of the hill plus the size of a certain slope extending one end (either the right or the left) of the hill. Such a slope cannot have size more than $|Q|^2$ because, otherwise, we obtain Condition (1). The size of the turn is therefore at most $|t_2-t_1|+|Q|^2\leq 3|Q|^3+|Q|^2\leq 4|Q|^3$. Similarly, the height of the turn is at most $|Q|^2 + |Q|^2 \leq 2|Q|^2$.

Moreover,  the inner states at the boundaries $t_1$ and $t_2$ must be  different because, otherwise, the corresponding state-stack pairs coincide, and thus Case 2 follows, a contradiction.

(b)
Next, we consider the case where $[s,s']$ consists of a series $\SSS$ of turns $[s_1,s_2], [s_2,s_3], \ldots, [s_{m-1},s_m]$ with $s=s_1$, $s'=s_m$, and $s_1<s_2<\cdots <s_m$. For each index $i\in[m-1]$, we take a hill $[t_i,t'_i]$ lying in $[s_i,s_{i+1}]$.
The height of such a hill is at most $|Q|^2$, and thus its boundary-block size is at most $3|Q|^3$ by (2a).
Recall that, for each turn $[s_i,s_{i+1}]$, its bottom height is the stack height at either end of the turn (i.e., $s_i$ or $s_{i+1}$).
For convenience, we define the \emph{gain} of the turn, denoted by $gain(s_i,s_{i+1})$, to be the stack height at the boundary $s_i$ minus the stack height at the boundary $s_{i+1}$.
It is easy to see that, when the gain of a turn equals $0$, this turn is merely a hill.
Let $(p_i,\gamma_i)$ denote the state-stack pair at the boundary $s_i$. Assume that the gains of all the turns are zero and that their bottom heights are all equal. In this special case, we obtain $m\leq |Q|$. This is shown as follows. Since the stack height does not go below the bottom height and $M$ is in an ideal shape, $\gamma_i = \gamma_j$ follows instantly for any pair $i,j\in[m-1]$. By Lemma \ref{mutual-matching}(1), if $m\geq |Q|+1$, then there are distinct indices $i,j\in[m]$ satisfying $p_i=p_j$, leading to Case 1, a contradiction. Therefore, $m\leq |Q|$ follows and $|s'-s|$ is upper-bounded by $m\cdot 4|Q|^3$, which is at most $4|Q|^4$.

(3) We estimate the total number of the turns whose bottom heights are a fixed number $h$. Given such a number $h$, we define $bh(h)$ to be the set of all indices $i\in[m-1]$ for which the bottom height of $[s_i,s_{i+1}]$ is exactly $h$. Note that $|bh(0)|\leq |Q|$ follows by (2b). We first assume that $|bh(h+1)|> |\bigcup_{i\in[0,h]_{\integer}} bh(i)|\cdot |Q|$ holds for a certain number $h\in\nat$. This implies the existence of more than $|Q|$ turns in $bh(h+1)$, any pair of which  contains no turn in $\bigcup_{i\in[0,h]_{\integer}}bh(i)$. Consider  $|Q|+1$ state-stack pairs at the boundaries whose heights match the bottom heights of those specific turns.  By Lemma \ref{mutual-matching}(1), at least two of those state-stack pairs coincide, a contradiction.
Hence, we conclude that, for any value $h\in\nat$, $|bh(h+1)|\leq  |\bigcup_{i\in[0,h]_{\integer}} bh(i)|\cdot |Q|$. We then assert that $|bh(h)|\leq 2^{h-1}|Q|^{h+1}$ for any $h\in\nat^{+}$. Since  $|bh(i)|\leq 2^{i-1}|Q|^{i+1}$ holds for any $i\in\nat^{+}$, by induction hypothesis, it follows that $|bh(h)|\leq
|Q|\cdot |\bigcup_{i\in[0,h-1]_{\integer}} bh(i)| \leq
|Q|(|Q|+\sum_{i=1}^{h-1}2^{i-1}|Q|^{i+1}) \leq (1+\sum_{i=1}^{h-1}2^{i-1})|Q|^{h+1} = 2^{h-1}|Q|^{h+1}$.

(4)
We define another notion of ``true gain'' in the $(s,s')$-region and intend to estimate its value.
Let us consider a series of consecutive turns $[s_1,s_2],[s_2,s_3],\ldots,[s_{m-1},s_m]$ in the boundary block $[s,s']$, where $s\leq s_1$, $s_m\leq s'$, and $m$ is chosen to maximize the length of this series.

For each turn, we consider its gain. Given an index $k\in[m]$, the \emph{true gain} in $[s_1,s_k]$, denoted by $tg(s_1,s_k)$, is set to be the sum $\sum_{i\in[k-1]}gain(s_i,s_{i+1})$.
For convenience, we also set $tg(s,s')=tg(s_1,s_m)$. We want to evaluate the value of this true gain and prove that $tg(s,s')>|Q|^3$. Assuming that $tg(s,s')\leq |Q|^3$, we wish to reach a contradiction. For this purpose, we examine two separate cases (a)--(b).

(a)
Firstly, we consider the case where $tg(s_1,s_k)\geq 2|Q|^3$ holds for a certain index $k\in[m-1]$.
We intend to partition all the turns according to the ``sign'' of their gains. For any pair $k,l\in[m]$ with $k<l$, the notation $P_{s_k,s_l}$ (resp., $N_{s_k,s_l}$ and $Z_{s_k,s_l}$) denotes the set of all indices $i\in[k,l-1]_{\integer}$ such that $gain(s_i,s_{i+1})>0$ (resp., $gain(s_i,s_{i+1})<0$ and $gain(s_i,s_{i+1})=0$).

Since $tg(s_1,s_k)\geq 2|Q|^3$, we can take elements $q$, $\alpha_1,\ldots,\alpha_{|Q|}$, and $\gamma$ satisfying that there exist  $|Q|+1$ boundaries in $\bigcup_{i\in P_{s_1,s_k}} [s_i,s_{i+1}]$ for which the state-stack pairs have the form  $(q,\gamma),(q,\alpha_1\gamma),\ldots,(q,\alpha_{[|Q|,1]}\gamma)$.
Matching these state-stack pairs,  since $tg(s,s')\leq |Q|^3$,
there are also state-stack pairs $(p_1,\alpha_{[|Q|,1]}\gamma), (p_2,\alpha_{[|Q|-1,1]}\gamma),\ldots, (p_{|Q|+1},\gamma)$ at  $|Q|+1$ boundaries in $\bigcup_{i\in N_{s_k,s_m}} [s_i,s_{i+1}]$ for properly chosen inner states $p_1,\ldots,p_{|Q|+1}$.
By Lemma \ref{mutual-matching}(1), we can find a distinct pair $i,j\in[|Q|+1]$ satisfying $p_i=p_j$. This implies Case 3, a contradiction.

(b) Consider the case where $tg(s_1,s_k)< 2|Q|^3$ holds for any index $k\in[|Q|+1]$.
The total number $m$ of turns in the $R$-region is at most $|\bigcup_{i\in[0,2|Q|^3]_{\integer}} bh(i)|$. Thus, we conclude that $m\leq \sum_{i=0}^{2|Q|^3}|bh(i) \leq |Q|+\sum_{i=1}^{2|Q|^3}2^{i-1}|Q|^{i+1} \leq |Q|+2^{2|Q|^3}|Q|^{2|Q|^3+1} \leq 2^{4|Q|^4}$.

By (2b), there may be series of at most $|Q|$ turns having gain $0$ and those turns have size at most $4|Q|^4$.
Since $|s'-s|$ is the sum of the sizes of all the turns, it follows that  $|s'-s|\leq |Q|\times 4|Q|^4 + m \cdot 4|Q|^4 \leq 4|Q|^42^{4|Q|^4} \leq 2^{5|Q|^5}$.
Since each turn has size at most $4|Q|^4$ and $m$ is the maximum value, it follows that $|s-(|x|-c)|<4|Q|^4$ and $||x|-s'|<4|Q|^4$. From these inequalities, we obtain $|s'-s|\geq c-8|Q|^4$.
We then conclude that $c\leq 2^{5|Q|^5}+8|Q|^4$. This contradicts the definition of $c$.

From the cases (a)--(b), we conclude that $tg(s,s')>|Q|^3$.

(5) Since $tg(s,s')\geq |Q|^3+1$, there exists a series $(q,\gamma), (q,\alpha_1\gamma),\ldots,(q,\alpha_{[|Q|^2,1]}\gamma)$ of $|Q|^2+1$ state-stack pairs at the boundaries $t_1,t_2,\ldots,t_{|Q|^2+1}$ in the $R$-region, respectively, for properly selected $q,\alpha_1,\ldots,\alpha_{|Q|^2+1}$. Consider the $y$-region of $x'y$. Let $r_1,r_2,\ldots,r_{|Q|^2+1}$ denote boundaries in the $y$-region
satisfying that the state-stack pairs at these boundaries have the form $(p_1,\alpha_{[|Q|^2,1]}\gamma), (p_2,\alpha_{[|Q|^2-1,1]}\gamma), \ldots,(p_{|Q|^2+1},\gamma)$, respectively. By Lemma \ref{mutual-matching}(2), there exists a subset $C\subseteq[|Q|^2+1]$ of size more than $|Q|$ satisfying $p_i=p_j$ for any pair $i,j\in C$. Next, we consider the $z$-region of $x'z$. Let $r'_1,r'_2,\ldots,r'_{|Q|+1}$ denote boundaries in the $z$-region for which the state-stack pairs at these boundaries have the form $(p'_1,\alpha'_{[|Q|,1]}), (p_2,\alpha'_{|Q|-1,1]}\gamma), \ldots, (p_{|Q|+1},\gamma)$, respectively. Lemma \ref{mutual-matching}(2) further implies $p'_{i_0}=p'_{j_0}$ for an appropriate distinct pair $i_0,j_0\in[|Q|+1]$.
This clearly implies the condition (2) of the lemma.

\subsection{Proof of Lemma \ref{pumping-lemma}}\label{sec:first-lemma-proof}

We are now ready to describe the proof of the first pumping lemma for $\dcfl[d]$ (Lemma \ref{pumping-lemma}). Our proof has two distinguished  parts depending on the value of $d$. The first part of the proof targets the basis case of $d=1$.
As argued in Section \ref{sec:new-tool}, this special case directly follows from Lemma \ref{pumping-lemma-state}.

Let $\Sigma$ be any alphabet and take  any infinite dcf language $L$ over $\Sigma$ in $\dcfl[d]$.
Our proof argument is easily extendable to \emph{one-way nondeterministic pushdown automata} (or 1npda's) and to the ``standard'' pumping lemma for $\cfl$ (cf. \cite{HU79}).
The second part of the proof deals with the general case of $d\geq2$.
Hereafter, we discuss these two parts separately.


\s
\n{\bf Basis Case of $d=1$:}
We begin our proof for the basis case of $d=1$. This comes from Lemma \ref{pumping-lemma-state} by taking an appropriate finite-state 1dpda $M$ for a given $L$ in Lemma \ref{pumping-lemma}. This completes the proof of the basis case of $d=1$.


\s
\n{\bf General Case of $d\geq2$:}
We begin our argument by considering $d$ dcf languages $L_1,L_2,\ldots,L_d$ satisfying $L=\bigcup_{i\in[d]}L_i$.
Take $d+1$ strings $w_1,w_2,\ldots,w_{d+1}$ of the form $xy^{(k)}$  in $L$ with $|x|>c$ for appropriate strings $x$ and $y^{(k)}$.
Since there are only $d$ languages for the $d+1$ strings, there must be an index $e\in[d]$ and two distinct indices $j_1,j_2\in[d+1]$ for which  $w_{j_1},w_{j_2}\in L_e$.
We then fix such a triplet $(j_1,j_2,e)$ and set $w=w_{j_1}$, $w'=w_{j_2}$, and $L'=L_e$  for simplicity.
Let us take any two factorizations $w=x'y$ and $w'=x'z$ with $|x'|>c$.  We apply the basis case of $d=1$ to $w$ and $w'$ and obtain one of the conditions (1)--(2) of the lemma. This completes the proof of the general case.

\section{Proof of the Second Pumping Lemma}\label{sec:second-lemma-verify}

We have already proven in Section \ref{sec:first-lemma-proof} the first pumping lemma for $\dcfl[d]$ by conducting a detailed analysis of a stack history of a given 1dpda. As announced in Section \ref{sec:new-tool}, by sharp contrast, this section  intends to prove the second pumping lemma for $\dcfl[d]$ (Lemma \ref{second-pumping-lemma}) by way of a grammar-based argument.

\subsection{Grammars and Grammatical Trees}\label{sec:grammatical-tree}

We begin with a brief explanation of crucial notions and notation regarding formal grammars. We will loosely follow the terminology of Yu \cite{Yu89} and, in part, that of Harrison and Havel \cite{HH74}.

A \emph{(formal) grammar} is a quadruple $G=\pair{N,\Sigma,P,S}$, where $N$ is a finite set of terminal symbols (or terminals, for short), $\Sigma$ is a finite set of non-terminal symbols (or nonterminals), $P$ is a set of production rules (or productions), and $S$ is the start symbol. We use the notation $\Rightarrow$ to denote a \emph{derivation} and $\Rightarrow_r$ to indicate a \emph{rightmost derivation}, in which any production should be first applied to the rightmost nonterminal. The notation $\Rightarrow_r^*$ indicates the reflexive, transitive closure of $\Rightarrow_r$.
Given any nonterminal $A\in N$ and any string $w\in (N\cup\Sigma)^*$, $w$ is called a \emph{right canonical grammatical form} of $A$ if $A\Rightarrow_r^* w$.
A grammar $G$ is said to be \emph{reduced} if (i) there is no reduction of the form $A\to A$ for any $A\in N$, (ii) for any $A\in N$, there is a string $w\in\Sigma^*$ satisfying $A\Rightarrow^{+} w$, and (iii) for any $A\in N$, there are strings $x,y\in (N\cup\Sigma)^*$ for which $S\Rightarrow^* xAy$.

A context-free grammar $G=\pair{N,\Sigma,P,S}$ is called an \emph{LR(1)  grammar} if (i) $S\not\Rightarrow_r^{+} S$, (ii) if $S\Rightarrow_r^* \alpha A w\Rightarrow_r \alpha\beta w = \gamma w$, $S\Rightarrow_r^* \alpha' A' u \Rightarrow_r \alpha'\beta' u = \gamma w'$, and $w_{(1)}=w'_{(1)}$, then $A=A'$, $\beta=\beta'$, $\alpha=\alpha'$, where $A,A'\in N$, $w,w',u\in \Sigma^*$, and $\alpha,\alpha',\beta,\beta',\gamma\in (N\cup\Sigma)^*$ (cf. \cite{HH74}).

Let $T$ denote a rooted tree with a set $D$ of nodes. For two nodes $v,w\in D$, if $v$ is a \emph{parent of} $w$ in $T$, then we write $v\triangleright w$. The \emph{dependency relation} $\triangleright^*$ is the reflexive and transitive closure of $\triangleright$.
Note that $v\triangleright^* w$ indicates the existence of a (directed) path from $v$ to  $w$ in $T$ and we also write $v\leadsto w$. In this case, $v$ is called an \emph{ancestor} of $w$.
Given a node $u$ in $T$, $Anc(u)$ denotes the set of all ancestors of $w$ in $T$; that is, $Anc(u) = \{w\in D\mid w\leadsto u\}$. In addition, we set $Anc^{+}(u) = Anc(u)-\{u\}$.
For two nodes $v,w\in D$, the notation $nca(v,w)$ denotes the \emph{nearest common ancestor} of $v$ and $w$; that is, the node $x$ in $D$ such that (i) $x\in Anc(v)\cap Anc(w)$ and (ii) no node $y$ in $Anc^{+}(v)\cap Anc^{+}(w)$  satisfies $x\leadsto y$.

We say that $u$ is a \emph{left to} $v$ in $T$, denoted $u\prec v$, if either (i) there exists a node $w\in D$ and an index $i\geq1$ such that $u$ is the $i$th child of $w$ and $v$ is the $(i+1)$th child of $w$ or (ii) there exist two nodes $w_1\in Anc(u)$ and $w_2\in Anc(v)$ satisfying $w_1\prec w_2$, and there is no node $x\in D$ satisfying $w_1\prec w_3\prec w_2$ for certain three nodes $w_1\in Anc(u)$, $w_2\in Anc(v)$, and
$w_3\in Anc(x)$. The \emph{left-to-right order} $\prec^*$ is the reflexive and transitive closure of $\prec$.
We write $leaf(T)$ for the sequence consisting of all leaves of $T$ in the left-to-right order.

The \emph{height} of a node $v$ in $T$ is the maximum length of a path from $v$ to a leaf in $T$ and the \emph{height} of $T$ is the height of its root.
Any node other than leaves is called an \emph{internal node}.
For an internal node $v$ in $T$, the subset $\{w\in D\mid w=v \text{ or } v\triangleright w\}$ of $D$ naturally induces a subtree of $T$ of height at most $1$. This subtree is called an \emph{elementary subtree of $T$ rooted at $v$}.
In contrast, the subset $\{w\in D\mid v\leadsto w\}$ for each node $v$ induces a subtree of $T$, which is denoted $T(v)$.

Associated with a grammar $G$, we deal with a labeled rooted tree $T$ with  a set $D$ of nodes, the root $r$ as well as two binary relations $\triangleright$ and $\prec$, and a labeling function $\rho: D\to N\cup\Sigma\cup\{\varepsilon\}$.
We further expand $\rho$ to a function from $D^*$ to $(V\cup\Sigma)^*$ by setting $\rho(v_1,v_2,\ldots,v_m) = \rho(v_1) \rho(v_2) \cdots \rho(v_m)$ for all nodes $v_1,v_2,\ldots,v_m\in D$, where $D^*$ denotes the set of all finite sequences of nodes in $D$.
We say that a production of the form $A\to\sigma_1\sigma_2\cdots \sigma_m$ (resp., $A\to \varepsilon$) for $A\in N$ and $\sigma_1,\sigma_2,\ldots,\sigma_m\in N\cup \Sigma$ \emph{expresses} an elementary subtree $T$ for which the root of $T$ has label $A$ and all children have the labels $\sigma_1,\sigma_2,\ldots,\sigma_m$ from left to right (resp., have label $\varepsilon$).
Such a tree $T$ is called a \emph{grammatical (derivation) tree} of $G$ if (i) for every elementary subtree $T'$ of $T$, there is a production in $P$ that expresses $T'$ and (ii) the string $\rho(leaf(T))$ belongs to $\Sigma^*$. Such a grammatical tree is expressed as $(D,\triangleright, \prec, r,\rho)$. When the root $r$ satisfies $\rho(r)=S$, $T$ is called a \emph{derivation tree} of $G$.
If $T$ further satisfies $\rho(r)\Rightarrow^*_r \rho(leaf(T))$ in $G$, then $T$ is called \emph{right canonical}.
It is often useful to relax the above condition (ii) to the following condition: (ii$'$) $\rho(leaf(T))\in (N\cup\Sigma)^*$. With this new condition (ii$'$) together with (i), we call $T$ a \emph{partial grammatical tree} of $G$. Let $\TT_{G}$ denote the set of all partial grammatical trees of $G$.
We say that $T$ is \emph{extendable} to grammatical trees if there exists a rightmost canonical grammatical tree $T'$ of $G$ such that $T$ is a subtree of $T'$ with the same root.

A leaf of $T$ whose label is a symbol in $\Sigma\cup\{\varepsilon\}$ is called a \emph{terminal node} and a single-node tree with a terminal node is succinctly called a \emph{terminal tree}.
Any node of $T$ whose label is
$\varepsilon$ is called an \emph{$\varepsilon$-node}.

Given a partial grammatical tree $T$ and a number $n\in\nat$, the \emph{left $n$-part}\footnote{This notion comes from Yu's work \cite{Yu89}, which is formulated quite differently from Harrison and Havel's \cite{HH74}. Nevertheless, the definition given in \cite{Yu89} is erroneous and thus requires attention.}
of $T$, denoted ${}^{(n)}T$, is a forest (i.e., a sequence of trees) defined as follows. Let ${}^{(0)}T=\setempty$. Assume that $leaf(T) = (v_1,v_2,\ldots,v_m)$ for $m\geq1$. If $n\leq m$, then ${}^{(n)}T$ is a subgraph of $T$, consisting of all nodes $z$ in $T$ such that, for a certain index $i\in[n]$, $z\prec^* v_i$ holds. In the case of $n>m$, we automatically set ${}^{(n)}T$ to be $T$.
Since ${}^{(n)}T$ is a sequence of subtrees of $T$, we assume that those subtrees are enumerated according to the left-to-right order.

For two partial grammatical trees $T_1 =(D_1,\triangleright_1, \prec_1, r_1,\rho_1)$ and $T_2 =(D_2,\triangleright_2, \prec_2, r_2,\rho_2)$, we say that $T_1$ is \emph{isomorphic} to $T_2$, denoted $T_1\equiv T_2$, if there exists a bijection $f:D_1\to D_2$ such that, for all $v,w\in D_1$,  (1) $f(r_1)= r_2$; (2) $v\triangleright_1 w$ iff $f(v)\triangleright_2 f(w)$; (3) $v\prec_1 w$ iff $f(v)\prec_2 f(w)$; and (4) $\rho_1(v)=\rho_2(f(v))$.
In this case, $f$ is called an \emph{isomorphism}. When $f(v_1)=v_2$, we also say that $v_1$ is \emph{isomorphic} to $v_2$.
Given  two (ordered) forests $F_1=(T_{11},T_{12},\ldots,T_{1k})$ and $F_2=(T_{21},T_{22},\ldots,T_{2m})$,
$F_1$ is said to be \emph{isomorphic} to $F_2$, denoted $F_1\equiv F_2$, if both $k=m$ and $T_{1i}\equiv T_{2i}$ hold for all indices $i\in[k]$.
For convenience, even if $F_1$ and $F_2$ are empty, we write $F_1\equiv F_2$ as well.

Let $v_1$ and $v_2$ denote two nodes of $T$ with $v_1\leadsto v_2$ and let $\alpha$ and $\beta$ be two sequences of leaves of $T$. We say that $(v_1,v_2)$ \emph{generates} $(\alpha,\beta)$ in $T$ if $leaf(T(v_1)) = \alpha \cdot leaf(T(v_2))\cdot \beta$. Here, the dot ``\,$\cdot$\,'' indicates the concatenation of two sequences. For example, $(v_1,v_2,v_3)\cdot (w_1,w_2)$ indicates $(v_1,v_2,v_3,w_1,w_2)$.
Notice that $\alpha$ or $\beta$ (or both) in $(\alpha,\beta)$ may be empty.
Consider a path $p$ from the root $r$ to a leaf in $T$. Given two distinct nodes $v_1$ and $v_2$ lying on the path $p$ with $v_1\leadsto v_2$, we define $st[v_1,v_2]$ to be the string pair $(\rho(\alpha),\rho(\beta))$ obtained from a pair $(\alpha,\beta)$ that is generated by $(v_1,v_2)$ in $T$. Moreover, let $ST[p]$ be the set $\{st[v_1,v_2]\mid v_1,v_2\in p, v_1\leadsto v_2, v_1\neq v_2\}$.


Yu \cite{Yu89} claimed the so-called \emph{left-part theorem for $\mathrm{LR}(k)$ grammars} (or the $\mathrm{LR}(k)$ left-part theorem).
Since every dcf language has an $\mathrm{LR}(1)$ grammar \emph{in Greibach normal form}\footnote{A grammar $G$ is \emph{in Greibach normal form} if every production has one of the following forms: $A\rightarrow a$ and $A\rightarrow bXY$ for $A,X,Y\in N$, $a\in\Sigma\cup\{\varepsilon\}$, and $b\in\Sigma$.} \cite{GHH76,Lom70}, for the proof of Lemma \ref{second-pumping-lemma}, we  require only the property of $\mathrm{LR}(1)$ grammars in Greibach normal form.
In this exposition, we therefore restrict our attention to the special case of $k=1$ and we rephrase Yu's original theorem as the following restricted form.

\begin{lemma}\label{left-part-theorem}
Let $G=\pair{N,\Sigma,P,S}$ denote any $\mathrm{LR}(1)$ grammar in Greibach normal form. Let $T_1=(D_1,\triangleright_1,\vdash_1,\gamma_1,\rho_1)$ and $T_2=(D_2,\triangleright_2,\vdash_2,\gamma_2,\rho_2)$ be two right canonical partial grammatical trees of $G$ with  $leaf(T_1) = (c_1,c_2,\ldots,c_m)$ and $leaf(T_2)= (d_1,d_2,\ldots,d_n)$ for $m,n\geq1$. If there exist two numbers  $s\in[0,m]_{\integer}$ and $t\in[0,n]_{\integer}$ satisfying the following six conditions, then both $s=t$ and ${}^{(s)}T_1 \equiv {}^{(t)}T_2$ hold.
(1) $\rho_1(r_1)=\rho_2(r_2)$.
(2) $\rho_1(c_1,c_2,\ldots,c_{s+1}) = \rho_2(d_1,d_2,\ldots,d_{t+1})\in(N\cup\Sigma)^*$.
(3) For all $i\in[s+1,m]_{\integer}$ and all $j\in[t+1,n]_{\integer}$ follows $\rho_1(c_i),\rho_2(d_j)\in \Sigma\cup\{\varepsilon\}$.
(4) $\rho_1(c_{s+1}) = \rho_2(d_{t+1})$.
(5) If $s<m$ and  $t<n$, then $\rho_1(c_{s+1})\in\Sigma$ and  $\rho_2(d_{t+1})\in\Sigma$.
(6) $T_1$ and $T_2$ are extendable to right canonical grammatical trees of $G$.
\end{lemma}

Figure \ref{fig:class-inclusion} illustrates
two right canonical partial grammatical trees discussed in Lemma \ref{left-part-theorem}.


\begin{figure}[t]
\centering
\includegraphics*[height=5.3cm]{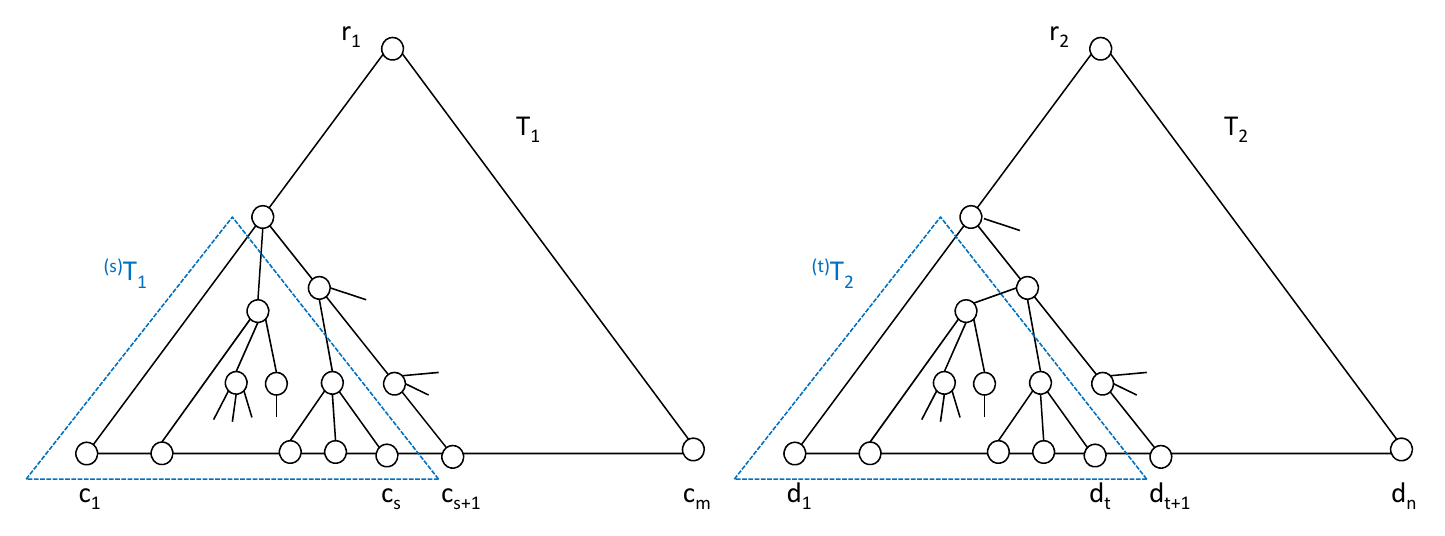}
\caption{An example of two right canonical partial grammatical trees given in Lemma \ref{left-part-theorem}.}\label{fig:class-inclusion}
\end{figure}


Since the proof of Yu's $\mathrm{LR}(k)$ left-part theorem is not available to general public, we need to  provide the proof of Lemma \ref{left-part-theorem} in order to prove the second pumping lemma.


\begin{proofof}{Lemma \ref{left-part-theorem}}
Let $G=\pair{N,\Sigma,P,S}$ be any $\mathrm{LR}(1)$ grammar in Greibach normal form and let $T_1= (D_1,\triangleright_1,{\prec_1}, r_1,\rho_1)$ and $T_2= (D_2,\triangleright_2,{\prec_2}, r_2,\rho_2)$ denote two right canonical partial grammatical trees of $G$. We assume that the conditions (1)--(6) of the premise of the lemma are satisfied by both $T_1$ and $T_2$.

Letting $leaf(T_1)=(c_1,c_2,\ldots,c_m)$ and $leaf(T_2)=(d_1,d_2,\ldots,d_n)$ with $m,n\geq1$, it follows by the premise of the lemma that   $\rho_1(c_i), \rho_2(d_j)\in N\cup\Sigma\cup\{\varepsilon\}$ for all $i\in[1,s]_{\integer}$ and all $j\in[1,t]_{\integer}$,
$\rho_1(c_{s+1}) = \rho_2(d_{t+1})\in\Sigma$, and
$\rho_1(c_i), \rho_2(d_j) \in\Sigma\cup\{\varepsilon\}$ for all $i\in[s+2,m]_{\integer}$ and all $j\in[t+2,n]_{\integer}$.
In particular, $S=\rho_1(r_1)$ and $\rho_1(c_1,\ldots,c_{s+1}) = \rho_2(d_1,\ldots,d_{t+1})$.


Consider the parents $v_0^{(1)}$ and $v_0^{(2)}$ of $c_{s+1}$ and $d_{t+1}$ in $T_1$ and $T_2$, respectively. We set $A_1=\rho_1(v_0^{(1)})$,  $a=\rho_1(c_{s+1})$, and $y=\rho_1(c_{s+2},\ldots,c_m)$.
Since $\rho_1(c_{s+1})\in\Sigma$, $v_0^{(1)}$ must have $c_{s+1}$ as its only child because $G$ is in Greibach normal form.
The corresponding production thus has the form $\rho_1(v_0^{(1)})\rightarrow \rho_1(c_{s+1})$, which is $A_1\rightarrow a$.
Similarly, $d_{t+1}$ is the only child of $v_0^{(2)}$.
Let $A_2=\rho_2(v_0^{(2)})$ and $z=\rho_2(d_{t+2},\ldots,d_n)$.
Notice that the corresponding production $\rho_2(v_0^{(2)})\rightarrow \rho_2(d_{t+1})$ is $A_2\rightarrow a$ since $\rho_2(d_{t+1})=a$.

Without loss of generality, we hereafter assume that $s\leq t$.
The following proof proceeds by induction on $s$.
We start with the basis case of $s=0$.
Each right canonical partial grammatical tree extendable to a right canonical grammatical tree can be viewed as a process of generating it by applying a series of rightmost derivations.

\begin{claim}\label{case-of-0}
If $s=0$, then $t=0$ follows. Therefore, ${}^{(s)}T_1={}^{(t)}T_2=\setempty$ holds.
\end{claim}

\begin{proof}
Assume that $s=0$. Toward a contradiction, we also assume that $t>0$.
It then follows that $\rho_1(c_1),\rho_2(d_{t+1})\in\Sigma$, $\rho_1(c_1)=\rho_2(d_1,\ldots,d_{t+1})$, and $\rho_1(c_1)=\rho_2(d_{t+1})$. From these conditions, we obtain $\rho_2(d_1)=\rho_2(d_2)=\cdots =\rho_2(d_t)=\varepsilon$.
The node $v_0^{(1)}$ is the parent of $c_1$ with $A_1=\rho_1(v_0^{(1)})$ and $a=\rho_1(c_1)$.
The last applied derivation for generating $T_1$ must be $A_1\rightarrow a$ since $c_1$ is the only child of $v_0^{(1)}$.
This implies that $S\Rightarrow_r^{*} A_1y \Rightarrow_r ay$.
In contrast, let us consider the parent $v'$ of $d_1$ and set  $A'=\rho_2(v')$. Notice that
$z=\rho_2(d_{t+1},\ldots,d_n)$. Since $\rho_2(d_1)=\varepsilon$, the last applied derivation for generating $T_2$ must be of the form $A'\rightarrow \rho_2(d_1)=\varepsilon$.
We then obtain $S\Rightarrow_r^* A'\rho_2(d_2,\ldots,d_{t+1}) z = A'az\Rightarrow_r az$.
By the LR(1) property of $G$, $A_1=A'$ and $a=\varepsilon$ follow. However, this is a clear contradiction with $a\in\Sigma$.
Therefore, we conclude that $t=0$.
\end{proof}

Next, we consider the case of $s>0$. This implies $t>0$ because, otherwise, Claim \ref{case-of-0} (by swapping between $s$ and $t$) yields $s=0$.

Take the maximum numbers $i_1\in[1,s]_{\integer}$ and $i_2\in[1,t]_{\integer}$ such that $c_{i_1}$ is a  leaf of $T_1$ whose label is also a non-terminal symbol and that $d_{i_2}$ is a leaf of $T_2$ whose label is a non-terminal symbol.
Let $v_1^{(1)}$ and $v_1^{(2)}$ respectively denote the parents of $c_{i_1}$ and $d_{i_2}$.
Since $\rho_1(c_{i_1})\in N$ and $\rho_2(d_{i_2})\in N$, the children of $v_1^{(1)}$ are expressed as $(c_{i_1-2},c_{i_1-1},c_{i_1})$ (from left to right) and the children of $v_1^{(2)}$ are $(d_{i_2-2},d_{i_2-1},d_{i_2})$ (from left to right).
Let $\beta'=\rho_1(c_{i_1-2},c_{i_1-1},c_{i_1})$ and $\beta''=\rho_2(d_{i_2-2},d_{i_2-1},d_{i_2})$.
Moreover, we obtain $\rho_1(c_{i_1-2}),\rho_2(d_{i_2-2})\in \Sigma$. For other nodes, we set $\alpha'=\rho_1(c_1,\ldots,c_{i_1-3})$ and $\alpha''=\rho_2(d_1,\ldots,d_{i_2-3})$ if $s>3$ and $t>3$.

\begin{claim}\label{N-cup-Sigma}
It follows that $\rho_1(c_j)\in N\cup \Sigma$ for all $j\in[i_1]$ and $\rho_2(d_j)\in N\cup\Sigma$ for all $j\in[i_2]$.
\end{claim}

\begin{proof}
We want to show only the first part of the claim because the second part can be similarly shown.
It suffices to prove that $\rho_1(c_j)\neq\varepsilon$. If $\rho_1(c_j)=\varepsilon$, then $c_j$ must have its parent, say, $u_j^{(1)}$  but it has no siblings. Let $B'_j=\rho_1(u^{(1)}_j)$.

Take an index $j\in[1,i_1]_{\integer}$ satisfying that $\rho_1(c_j)=\varepsilon$ and that the production $\rho_1(u_j^{(1)})\rightarrow \rho_1(c_j)$ (i.e., $B'_j\rightarrow \varepsilon$) is the most recently applied one while generating  $T_1$. Since $T_1$ is right canonical, the production $\rho_1(u_j^{(1)})\rightarrow \varepsilon$ cannot be applied earlier than the production $\rho_1(v_1^{(1)})\rightarrow  \rho_1(c_{i_1-2},c_{i_1-1},c_{i_1})$ (i.e., $B_1\rightarrow \beta'$).
Even after the application of $B_1 \rightarrow \beta'$,
we cannot apply $B'_j \rightarrow \varepsilon$ because $\rho_1(c_{i_1})\in N$. Hence, there is no chance of applying $B'_j \rightarrow \varepsilon$. This is a contradiction. Therefore, $\rho_1(c_j)$ cannot be $\varepsilon$.
\end{proof}

Recall that $1\leq i_1\leq s$ and $1\leq i_2\leq t$. Since $\rho_1(c_1,\ldots,c_{s+1}) = \rho_2(d_1,\ldots,d_{t+1})$, we derive $i_1=i_2$ from Claim \ref{N-cup-Sigma}. It thus follows that $\rho_1(c_j)=\rho_2(d_j)$ for all indices $j\in[i_1]$.

In what follows, let us consider two separate cases. The first case is $i_1=s$ and the second case is $i_1<s$.

(1) We intend to discuss the first case of $i_1=s$.

\begin{claim}\label{value-equal-s}
If $i_1=s$, then $i_2=t$ follows.
\end{claim}

\begin{proof}
Assume that $i_1=s$. Toward a contradiction, we further assume that $i_2<t$.
Since $c_s = c_{i_1}$ with $\rho_1(c_{i_1})\in N$, the children  $(c_{s-2},c_{s-1},c_s)$ of $v_1^{(1)}$ satisfy that $\rho_1(c_{s-2})\in \Sigma$ and  $\rho_1(c_{s-1}),\rho_1(c_{s})\in N$ because $G$ is in Greibach normal form.
This corresponds to the production of the form $B_1\rightarrow \beta'$, which is the last production to have been applied to generate $T_1$.
Now, we delete the node $v_1^{(1)}$ from $T_1$ and call the obtained subtree by $T'_1$. Recall that $\alpha' = \rho_1(c_1,\ldots,c_{i_1-3})$.
We then obtain $S\Rightarrow_r^* \rho_1(c_1,\ldots,c_{s-3}) \rho_1(v_1^{(1)}) \rho_1(c_{s+1},\ldots,c_m) = \alpha' B_1 ay \Rightarrow_r \alpha' \beta' a y$.

Let us look into $T_2$. Since $d_{i_2}$ is the rightmost leaf labeled with a nonterminal, we conclude that $\rho_2(d_j)=\varepsilon$ for all indices $j\in[i_2+1,t]_{\integer}$.
This means that, for each index $j\in[i_2+1,t]_{\integer}$, there is a node  $u^{(2)}_{j}$ for which $d_j$ is the only child of $u^{(2)}_j$.
For brevity, we write $B''_j$ for $\rho_2(u^{(2)}_j)$.
Its associated production $\rho_2(u^{(2)}_j)\rightarrow \rho_2(d_j)$ is $B''_j \rightarrow \varepsilon$.
We delete $d_{i_2+1}$ from $T_2$ and call the obtained subtree by $T'_2$. From this $T'_2$, we obtain $S\Rightarrow_r^* \rho_2(d_1,\ldots,d_{i_2}) \rho_2(u^{(2)}_{i_2+1}) \rho_2(d_{i_2+2},\ldots,d_{t}) \rho_2(d_{t+1},\ldots,d_n) = \alpha'' \beta'' B'_{i_2+1} az \Rightarrow_r \alpha'' \beta'' az$. By the LR(1) property of $G$, it then follows that $B_1 = B''_{i_2+1}$ and $\beta'=\varepsilon$. This is a contradiction because of  $\beta'=\rho_1(c_{s-2},c_{s-1},c_s)$. Therefore, $i_2=t$ must follow.
\end{proof}


\begin{figure}[t]
\centering
\includegraphics*[height=4.2cm]{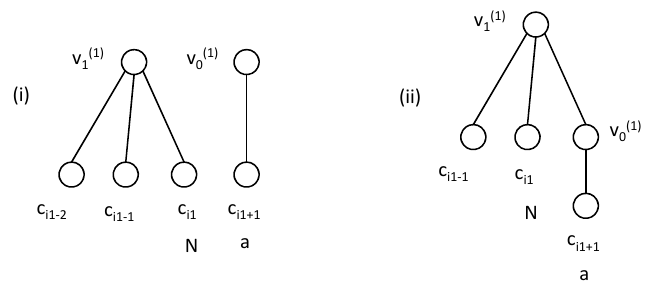}
\caption{Graphical descriptions of the cases (i) and (ii)}\label{fig:case-study}
\end{figure}


Note that, if $i_1=i_2$, then Claim \ref{value-equal-s} immediately yields $s=t$. In what follows, we focus on $v_1^{(1)}$ and $v_1^{(2)}$.
Recall that $B_1=\rho_1(v_1^{(1)})$ and $B_2=\rho_2(v_1^{(2)})$.

(i) Consider the case where $(c_{i_1-2},c_{i_1-1},c_{i_1})$ and $(d_{i_2-2},d_{i_2-1},d_{i_2})$ are children of $v_1^{(1)}$ and $v_1^{(2)}$, respectively. See Fig.~\ref{fig:case-study}(i).
Since $i_1=s$, $\rho_1(c_{i_1+1})=a$ holds.
We delete all children of $v_1^{(1)}$ from $T_1$ and all children of $v_1^{(2)}$ from $T_2$. We respectively  call the obtained subtrees by $T''_1$ and $T''_2$.
Note that $leaf(T''_1) = (c_1,\ldots,c_{i_1-3},v_1^{(1)},c_{i_1+1},\ldots,c_m)$ with $|leaf(T''_1)|+2=|leaf(T_1)|$.
We then obtain the derivations $S\Rightarrow_r^* \alpha' \rho_1(v_1^{(1)}) ay = \alpha' B_1 ay \Rightarrow_r \alpha'\beta' ay$ by applying the production $B_1\rightarrow \beta'$.
Similarly, since $leaf(T''_2) = (d_1,\ldots,d_{i_2-3},v_1^{(2)},d_{i_2+1},\ldots,d_n)$ with $|leaf(T''_2)|+2=|leaf(T_2)|$, we obtain $S\Rightarrow_r^* \alpha'' \rho_2(v_1^{(2)}) az = \alpha'' B_2 az \Rightarrow_r \alpha''\beta'' az$ by applying $B_2\rightarrow \beta''$.
The LR(1) property of $G$ then yields both $B_1=B_2$ and $\beta'= \beta''$.

Let $F_1$ and $F_2$ respectively denote the single-node trees consisting of $v_1^{(1)}$ and $v_1^{(2)}$ and let $F'_1$ and $F'_2$ respectively denote the subtrees of $T_1$ and $T_2$, which are rooted at $v_1^{(1)}$ and $v_1^{(2)}$. Let ${}^{(s-2)}T''_1 = (F_{11},\ldots,F_{1l_1})$ and ${}^{(t-2)}T''_2 =(F_{21},\ldots,F_{2l_2})$. By the definitions of $T''_1$ and $T''_2$, it follows that $F_{1l_1}=F_1$ and $F_{2l_2}=F_2$.
Therefore, $T''_1$ and $T''_2$ satisfy the conditions (1)--(5) of the lemma.
Since $|leaf(T''_1)|<|leaf(T_1)|$ and $|leaf(T''_2)|<|leaf(T_2)|$, we apply induction hypothesis to $T''_1$ and $T''_2$ and then obtain $s-2=t-2$ and ${}^{(s-2)}T''_1\equiv {}^{(t-2)}T''_2$. It thus follows that $l_1=l_2$ and $F_{1i}\equiv F_{2i}$ for all $i\in[l_1]$. By replacing $F_1$ and $F_2$ respectively with $F'_1$ and $F'_2$, we obtain ${}^{(s)}T_1 = (F_{11},\ldots,F_{1l_1-1},F'_1)$ and ${}^{(t)}T_2 = (F_{21},\ldots,F_{2l_2-1},F'_2)$.
Since $F'_1\equiv F'_2$ and $l_1=l_2$, we conclude that ${}^{(s)}T_1\equiv {}^{(t)}T_2$.

(ii) Consider the case where the children of $v_1^{(1)}$ are $(c_{i_1-1},c_{i_1},v_0^{(1)})$ with $\rho_1(c_{i_1+1})=a$ and the children of $v_1^{(2)}$ are $(d_{i_2-1},d_{i_2},v_0^{(2)})$. See Fig.~\ref{fig:case-study}(ii).
Let $X=\rho_1(c_{i_1})$. There are two grammatical trees $\tilde{T}_1$ and $\tilde{T}_2$ for which  their roots have label $X$ and all leaves of them are labeled by terminal symbols. We then expand $T_1$ and $T_2$ by attaching $\tilde{T}_1$ and $\tilde{T}_2$ to their leaves $c_{i_1}$ and $d_{i_2}$, respectively, and we call the resulting trees by $T_1^{(ext)}$ and $T_2^{(ext)}$. We set $s'=i_1-1$ and $t'=i_2-1$. The newly constructed trees $T_1^{(ext)}$ and $T_2^{(ext)}$ satisfy the conditions (1)--(6) of the lemma.

(iii) We study the case where  $(c_{i_1-2},c_{i_1-1},c_{i_1})$ are the children of $v_1^{(1)}$ but $(d_{i_2-1},d_{i_2},v_0^{(2)})$ are the children of $v_1^{(2)}$. In this case, we obtain the derivations $S\Rightarrow_r^* \rho_1(c_1,\ldots,c_{i_1-3}) \rho_1(v_1^{(1)}) \rho_1(c_{i_1+1},\ldots,c_m) = \alpha' B_1ay \Rightarrow_r \alpha'\beta' ay$ by applying the production
$B_1\rightarrow \beta'$ and also
$S\Rightarrow_r^* \rho_2(d_1,\ldots,d_{i_2}) \rho_2(v_0^{(2)}) \rho_2(d_{i_2+2},\ldots,d_n) = \alpha'' \beta'' A_2z \Rightarrow_r \alpha''\beta' az$ by applying  $A_2\rightarrow a$. The LR(1) property of $G$ then concludes that $B_1=A_2$ and $\beta'=a$. This is a contradiction with $|\beta'|=3$. Therefore, this case never happens.


(2) We next examine the second case of $i_1<s$. We set $p= s-i_1$ so that $c_{i_1+p} = c_s$. This implies that  $\rho_1(c_{i_1+1},\ldots,c_{s})=\varepsilon$. Let $u^{(1)}_j$ denote the parent of $c_j$. Notice that $c_j$ is the only child of $u^{(1)}_j$.
Similarly, we set $p'=t-i_2$. By Claim \ref{value-equal-s} (by swapping $T_1$ and $T_2$), $p'$ cannot be $0$, and thus we obtain $p'>0$. We also write  $u^{(2)}_j$ for the parent of $d_j$. We delete $c_{i_1}$ from $T_1$ and $d_{i_2+1}$ from $T_2$ and we then call the resulting subtrees by $T'_1$ and $T'_2$, respectively.
 For simplicity, let $A'_j=\rho_1(u^{(1)}_{j})$ and $A''_j=\rho_2(u^{(2)}_{j})$.

\begin{claim}\label{p-equal-p}
If $i_1<s$ then $p=p'$ follows.
\end{claim}

\begin{proof}
To lead to a contradiction, we first assume that $p<p'$. Since $\rho_1(c_{i_1+1})=\varepsilon$ and $\rho_1(c_{i_1})\in N$, the last production applied to produce $T'_1$ must be $B'_{i_1+1}\rightarrow \varepsilon$. From $T'_1$, we thus obtain the derivations
$S\Rightarrow_r^* \rho_1(c_1,\ldots,c_{i_1}) \rho_1(u^{(1)}_{i_1+1}) \rho_1(c_{i_1+2},\ldots, c_{s}) \rho_1(c_{s+1},\ldots,c_m) = \alpha' \beta' A'_{i_1+1} ay \Rightarrow_r \alpha'\beta' ay$. Similarly, from $T'_2$, we obtain $S\Rightarrow_r^* \rho_2(d_1,\ldots,d_{i_2}) \rho_2(u^{(2)}_{i_2+1}) \rho_2(d_{i_2+2},\ldots,d_t) \rho_2(d_{t+1},\ldots,d_n) = \alpha''\beta'' A''_{i_2+1} az \Rightarrow_r \alpha''\beta'' az$ by applying  $A''_{i_2+1} \rightarrow \varepsilon$.
The LR(1) property of $G$ then ensures that $A'_{i_1+1}= A''_{i_2+1}$.
We further delete $c_{i_1+2}$ from $T'_1$ and $d_{i_2+2}$ from $T'_2$. A similar argument as above concludes that $A'_{i_1+2} = A''_{i_2+2}$. Repeating this process eventually leads to $A'_{i_1+k} = A''_{i_2+k}$ for all $k\in[1,p]_{\integer}$.
Finally, since $p<p'$, it is possible to delete $c_{i_1}$ from $T_1$ and $d_{i_2+p+1}$ from $T_2$.
We then obtain $S\Rightarrow_r^* \alpha' \rho_1(v_1^{(0)}) \rho_1(c_{i_1+1},\ldots,c_{s}) ay = A_1 ay \Rightarrow_r \alpha'\beta' ay$ and $S\Rightarrow_r^* \alpha'' \beta'' \rho_2(d_{i_2+1},\ldots,d_{i_2+p}) \rho_2(u^{(2)}_{i_2+p+1}) az = \alpha''\beta'' A''_{i_2+p+1} az \Rightarrow_r \alpha''\beta'' az$.
From these, we conclude that $A_1=A''_{i_2+p+1}$ and $\beta' = \varepsilon$. This is a contradiction with $\beta'\neq\varepsilon$, and thus  $p\geq p'$ follows. By symmetry, when $p >p'$, we also obtain a contradiction. Therefore, we conclude that $p=p'$.
\end{proof}

In  this case, we ignore the nodes in $(c_{i_1+1},\ldots,c_{i_1+p})$ and $(d_{i_2+1},\ldots,d_{i_2+p'})$ and then take the same argument as in (ii) to prove the lemma.
\end{proofof}

\subsection{Proof of Lemma \ref{second-pumping-lemma}}

To prove Lemma \ref{second-pumping-lemma}, let us take any integer $d\geq1$ and fix an arbitrary infinite language $L$ in $\dcfl[d]$. We split the intended proof of the lemma into two parts, depending on the value of $d$.

\ms
\n{\bf Proof for the basis case of $d=1$:}
\s

Let us first consider the basis case of $d=1$.
Consider an $\mathrm{LR}(1)$ grammar $G=\pair{N,\Sigma,P,S}$ for $L$ and assume that $G$ already becomes  reduced. As noted in Section \ref{sec:grammatical-tree}, we further assume without loss of generality that $G$ is in Greibach normal form.
Let $k=|N|$ and let $h$ denote the maximum length of the right-hand side of any production of $G$. We define $c=2^{k\log{h}}(k+1)(h-1)$.

Take two arbitrary strings $w=xy$ and $w'=xz$ in $L$ with $|x|,|y|,|z|\geq1$ and $|x|>c$. Let $a$ denote the last symbol of $x$ and let $x^{(-)}$ denote the string obtained from $x$ by deleting its last symbol. Note that $x=x^{(-)}a$ and $|x^{(-)}|\geq c$.

Let $T_w=(N_1,\triangleright_1,\prec_1,r_1,\rho_1)$  and $T_{w'}=(N_2,\triangleright_2,\prec_2,r_2,\rho_2)$ denote two arbitrary right canonical grammatical (derivation) trees producing $w$ and $w'$, respectively.
Let $leaf(T_w)=(c_1,c_2,\ldots,c_m)$ and $leaf(T_{w'})= (d_1,d_2,\ldots,d_l)$ (in the left-to-right order) for certain numbers $m,l\geq1$.

Let $xy = a_1a_2\cdots a_{l}$ with $l=|xy|$ and $a_i\in\Sigma$ for all $i\in[l]$.
Fix $i\in[l]$ arbitrarily and take a number $s_i\in[0,m-1]_{\integer}$ for which $\rho_1(c_1,c_2,\ldots,c_{s_i}) =a_1a_2\cdots a_{i-1}$ and $\rho_1(c_{s_{i}+1}) =  a_{i}$.
Similarly, for $xz=a'_1a'_2\cdots a'_{l'}$, we take a number $t_i\in[0,n-1]_{\integer}$ for which $\rho_2(d_1,d_2,\ldots,d_{t_i})=a'_1a'_2\cdots a'_{i-1}$ and $\rho_2(d_{t_i+1})= a'_{i}$.

Let us consider a path $p_i$ of $T_w$ from $r_1$ to $c_{s_i+1}$ and another path $p'_i$ of $T_{w'}$ from $r_2$ to $d_{t_i+1}$.
Note that $p_i$ splits $T_w$ into two parts ${}^{(s_i)}T_{w}$
and $F_1$ and that $p'_i$ splits $T_{w'}$ into ${}^{(t_i)}T_{w'}$ and $F_2$
for two appropriate subtrees $F_1$ and $F_2$.
When $i\geq |x^{(-)}|$, these pathes $p_i$ and $p'_i$ may be even structurally different, and thus they may not be ``identical''.
Let ${}^{(s_i)}T_w=(T_1,T_2,\ldots,T_{m_i})$ whose subtrees are enumerated according to the left-to-right order.
Let $i_0$ and $i'_0$ satisfy that $\rho_1(c_1,c_2,\ldots,c_{t_{i_0}})= \rho_2(d_1,d_2,\ldots,d_{t_{i'_0}})= x^{(-)}$ with $\rho_1(c_{t_{i_0}+1})= \rho_2(d_{t_{i'_0}+1}) \in\Sigma$.

For convenience, given a path $p$ and a symbol $A\in N$, we introduce the notation $node_{p}[A]$ to indicate the set of all pairs $(v_1,v_2)$ of distinct nodes on the path $p$ for which $v_1\leadsto v_2$ and $v_1$ and $v_2$ have the same label  $A$.


We first argue that there are iterative pairs and later discuss various situations of these iterative pairs.

We split $xy$ to two parts: a prefix $\hat{x}$ and a suffix $\hat{y}$ of $xy$ satisfying $\hat{x}\hat{y}=xy$. Let $\hat{a}$ denote the first symbol of $\hat{y}$ if $\hat{y}\neq\varepsilon$. Along the path $p_i$ from the root $r_1$ to $c_{s_i+1}$, $\rho_1(c_{s_i+1})=\hat{a}$ follows.

(1) Consider the case where there exist two indices $i\in[l]$ and $j\in[m_i]$ for which a subtree $T_{j}$ of ${}^{(s_i)}T_w$ has at least $2^{k\log{h}}$ leaves.
For such a subtree $T_j$, we claim the existence of a path $p'$ of $T_j$ from the root of $T_j$ to a leaf that has length at least $k$ because, otherwise, the height of $T_j$ is at most $k-1$, and thus there are at most $h^{k-1}$ ($=2^{(k-1)\log{h}}$) leaves, a contradiction with the choice of $T_j$.  Since the path $p'$ has at least $k+1$ nodes and $k=|N|$, there exist two distinct nodes $v_1$ and $v_2$ on the path $p'$ whose labels are the same.
Let $\hat{x}=a_1a_2\cdots a_{i-1}$, $\hat{a}=a_i$, and $\hat{y}=a_{i+1}\cdots a_l$.
Take two strings $f,g\in\Sigma^+$ for which $st[v_1,v_2]=(f,g)$. Clearly, $f\sqsubseteq \hat{x}$ and $g\sqsubseteq \hat{x}$ hold.

(a) Assume that $\hat{x}\sqsubseteq x^{(-)}$. In this case, we obtain $f,g\sqsubseteq x^{(-)}$. Since $c_{i} \prec^{+} c_{s_{i_0}}$,
we apply Lemma \ref{left-part-theorem} and then obtain $s_i=t_i$ and  ${}^{(s_i)}T_{w}\equiv {}^{(t_i)}T_{w'}$.
From this fact,  hereafter, we can focus only on ${}^{(s_i)}T_{w}$.
Assume that $(v_1,v_2)$ generates $(\beta_1,\beta_2)$ in $T_{w}$ for two certain sequences $\beta_1$ and $\beta_2$ of leaves in $T_w$.
Moreover, assume that $(r_1,v_1)$ generates $(\alpha_1,\alpha_2)$ in $T_{w}$ for two certain sequences $\alpha_1$ and $\alpha_2$ of leaves.
It then follows from $T_w=T_w(r_1)$ that $leaf(T_{w}) = \alpha_1 \cdot leaf(T_{w}(v_1)) \cdot \alpha_2$ and $leaf(T_{w}(v_1)) = \beta_1 \cdot leaf(T_{w}(v_2)) \cdot \beta_2$.
Notice that $\rho_1(\beta_1) =f$ and $\rho_1(\beta_2)=g$.
We then set $x_1=\rho_1(\alpha_1)$, $x_2= f$, $x_3=\rho_1(leaf(T_{w}(v_2))$, $x_4= g$, and $x_5y = \rho_1(\alpha_2)$. It is easy to check that the condition (1) of the Lemma \ref{second-pumping-lemma} is fulfilled.

(b) Assume that $f\sqsubseteq ay$ and $g\sqsubseteq az$. This case implies the condition (2).

(c) Assume that $f\sqsubseteq x^{(-)}$ and $g\sqsubseteq ay$. In this case, we obtain the condition (3).

(2) Consider the case where, for any $i$ and any $j\in[m_i]$, the subtree $T_j$ in ${}^{(s_{i})}T_w$ has less than $2^{k\log{h}}$ leaves.
Let $\hat{x}=a_1a_2\cdots a_{i-1}$ and $\hat{a}=a_i$ and assume that $|\hat{x}|\geq c$. There are more than $(k+1)(h-1)$ distinct subtrees in ${}^{(s_i)}T_w$ because, otherwise, there are less than $2^{k\log{h}}(k+1)(h-1)$ leaves in ${}^{(s_i)}T_w$, and thus $|\hat{x}|< c$ follows, an immediate  contradiction.

We then focus on the last $(k+1)(h-1)+1$ subtrees in ${}^{(s_i)}T_w$. Consider the path $p_i$ from $r_1$ to $c_{s_i+1}$.
Each internal node on this path $p_i$ has at most $h$ children and it is a parent of one terminal node since $G$ is in Greibach normal form.
We claim that the path $p_i$ has at least $k+1$ nodes. This is because, otherwise, each internal node on the path $p_i$ has at most $h-1$ direct descendants, and thus ${}^{(s_i)}T_w$ has at most $(k+1)(h-1)$ subtrees, leading to a contradiction. Since $k=|N|$, the pigeonhole principle concludes that there are two nodes, say, $v_1$ and $v_2$ on the path $p_i$ having the same labels. Let $st[v_1,v_2]=(f,g)$. Any node lying  between $v_1$ and $v_2$ on the path $p_i$ is a root of a subtree of at most $2^{k\log{h}}$ leaves.
From this fact, we conclude that $|f|$ is upper-bounded by $(k+1)2^{k\log{h}}$, and therefore $|f|\leq c$ follows. Note that $f\sqsubseteq \hat{x}$ and $g\sqsubseteq \hat{y}$.

Assume that $i\leq i_0$. In this case, $c_{s_i+1}\prec^{+} c_{s+1}$ follows, and thus $v_1$ and $v_2$ are in a subtree whose root is a child of a certain node on the path $p_i$. By Lemma \ref{left-part-theorem},  ${}^{(t_i)}T_{w'}$ is isomorphic to ${}^{(s_i)}T_w$, and thus there are two nodes $v'_1$ and $v'_2$ in that subtree respectively isomorphic to $v_1$ and $v_2$. Note that $\rho_1(v_1)=\rho_1(v_2)$ and $\rho_2(v'_1) = \rho_2(v'_2)$.

We begin with fixing $v_1$ and $v_2$ so that the height of $v_2$ is the smallest.
For readability, let $x'=x_1x_2x_3x_4x_5x_6x_7$, $y=u_1u_2u_3$, and $z=u'_1u'_2u'_3$ and assume that $st[v_1,v_2]$ is expressed as $(x_4x_5,x_7u_1)$.
To lead to the condition (5), we need to show that $st[v'_1,v'_2]$ has the form $(x_3x_4,u'_2)$ or $(x_2,u'_2)$ by reassigning $x_2$, $x_3$, and $x_4$ appropriately.
Toward this goal, it suffices to verify that $st[v'_1,v'_2]$ never has the form: $(x_5x_6,u'_2)$, $(x_6,u'_2)$, and $(x_7,u'_2)$ (after an appropriate rearrangement of $x_5,x_6,x_7,u'_2$).
Assume that there is no $v$ for which $v_1\leadsto v\leadsto v_2$ with $\rho_1(v)=\rho_1(v)$ and that there is no $v'$ for which $v'_1\leadsto v'\leadsto v'_2$ with $\rho_2(v')=\rho_2(v'_1)$.

Toward a contradiction, we first assume that $st[v'_1,v'_2]=(x_5x_6,u'_2)$.
Let $\hat{a}$ denote the last symbol of $x_7$ and let $x_7^{(-)}$ denote the string satisfying $x_7=x_7^{(-)}\hat{a}$. Take a leaf $c_{s_{i_1+1}}$ of $T_w$ whose  label is $\hat{a}$ and consider a path of $T_w$ from $r_1$ to $c_{i_1+1}$. Since $st[v_1,v_2]=(x_4x_5,x_7u_1)$, $c_{i_1+1}\prec^{+} c_{i_0+1}$ follows. By Lemma \ref{left-part-theorem}, the subtrees of ${}^{(s_{i_0})}T_w$ and those of ${}^{(t_{i'_0+1})}T_{w'}$ are isomorphic. There exists a node $v'$ in $T_{w'}$ such that $v'$ is isomorphic to $v_2$ in $T_w$. This implies that $v'_1\leadsto v'\leadsto v'_2$. Hence, $T_{w'}(v'_2)$ cannot contain nodes producing $x_7^{(-)}$. Since $st[v'_1,v'_2]=(x_5x_6,u'_2)$, $T_{w'}(v'_2)$ contains nodes producing $x_7$, a contradiction.

The other cases of $st[v'_1,v'_2]\in \{(x_6,u'_2),(x_7,u'_2)\}$ are similarly handled.


\ms
\n{\bf Proof for the general case of $d\geq2$:}
\s

Let us prove the lemma for a general case of $d\geq2$. Take $d+1$ strings $w_1,w_2,\ldots,w_{d+1}$ in $L$ such that there are strings $x, y^{(1)},\ldots,y^{(d+1)}$ for which $w_j$ has the form $xy^{(j)}$ for any $j\in[d+1]$. Since $L\in\dcfl[d]$, there are $d$ dcf languages $L_1,L_2,\ldots,L_d$ satisfying $L=\bigcup_{i\in[d]} L_i$.
By the pigeonhole principle, there exist an index $i\in[d]$ and two distinct indices $j_1,j_2\in[d+1]$ satisfying $w_{j_1},w_{j_2}\in L_i$. We take strings $x'$, $y$, and $z$ with $|x'|>c$ satisfying that $x'y=xy^{(j_1)}$ and $x'z=xy^{(j_2)}$. We then apply the basis case of $d=1$ and obtain one of the conditions (1)--(5).


\section{A Brief Discussion and Future Challenges}

We have presented two new and practical lemmas, referred to as the \emph{pumping lemmas for $\dcfl[d]$} in this exposition.
The first one (Lemma \ref{pumping-lemma}) is in part viewed as a natural extension of Yu's pumping lemma for $\dcfl$ \cite{Yu89}. This lemma has a quite simple form, describing a structural property of iterative pairs. In Section \ref{sec:proof-separation}, we have applied it to separate the intersection and the union hierarchies of deterministic context-free (dcf) languages. Our proof of this pumping lemma is solely founded on an analysis of the behaviors of one-way deterministic pushdown automata (or 1dpda's) and it therefore provides an alternative proof to Yu's pumping lemma whose proof actually utilizes $\mathrm{LR}(k)$ grammars.
In contrast to the first pumping lemma, the second pumping lemma for $\dcfl[d]$ (Lemma \ref{second-pumping-lemma}) relates to a structural property of nondegenerate iterative pairs. The proof of this pumping lemma
is made along the line of a detailed analysis of grammatical (derivation) trees of $\mathrm{LR}(k)$ grammars. As a future task, it is desirable to conduct a study on various aspects of iterative pairs particularly for dcf languages.

Although our pumping lemmas are quite useful, it is not as powerful as to precisely \emph{characterize} all $d$-union dcf languages.
One of the most challenging tasks is to expand our pumping lemmas to characterize all dcf languages. See, e.g., \cite{Wis76} for the case of context-free languages.

Since finite intersections and unions have been intensively discussed so far, we here wish to discuss their simple extension to ``infinite''  intersections in a certain ``controlled'' way. Given a 1dpda $M= (Q,\Sigma,\{\cent,\dollar\}, \Gamma, \delta,q_0,Z_0, Q_{acc}, Q_{rej})$, we first assign to $M$ its \emph{description size}, which equals the value $|Q||\Sigma||\Gamma^{\leq e}|$. We express this value as $des(M)$.
Let us consider an infinite sequence $\{M_n\}_{n\in\nat}$ of 1dpda's. Such a sequence is said to have \emph{polynomial description size} if there exists a polynomial (with nonnegative coefficients) $p$ such that $des(M_n)\leq p(n)$ holds for any number $n\in\nat$. Given a function $\mu:\nat\to\nat$, we concentrate on the \emph{$\mu$-bounded intersection} of the dcf languages $\{L(M_n)\}_{n\in\nat}$, which is defined as $\{x\mid \forall i\leq \mu(|x|)[x\in L(M_i)]\}$.
Although the language $Pal$ of even-length palindromes is not in $\dcfl[\omega]$ by Theorem \ref{Pal-DCFL}(1), it is possible to characterize it in terms of a certain $\mu$-bounded intersection of dcf languages as follows.

\begin{lemma}\label{mu-bound}
There exists an infinite sequence $\{M_n\}_{n\in\nat}$ of 1dpda's having polynomial description size such that $Pal$ coincides with the $\mu$-bounded intersection of $\{L(M_n)\}_{n\in\nat}$, where $\mu(n)=\ceilings{n/2}$ for any $n\in\nat$.
\end{lemma}

\begin{proof}
We wish to construct a desired family $\MM=\{M_n\}_{n\in\nat}$ of 1dpda's.
Let $M_0$ denote a 1dpda that recognizes $\{x\mid |x|\text{ is even }\}$ without using its stack. For any other index $n\geq1$, we define $M_n$ in the following way. Given an input $x$, if $|x|\leq n$, then we reject $x$. Assume otherwise. We push all symbols of $x$ one by one into a stack and remember the $n$th symbol of $x$, say, $a_n$ in the form of inner states. This is possible because $n$ is treated as a ``constant'' for $M_n$, not depending on the length of the input $x$. After reaching the endmarker $\dollar$, we start to pop the last $n$ symbols of $x$ from the stack. We then check if the $n$th symbol matches $a_n$. If so, then we accept $x$; otherwise, we reject $x$. Consider the $\mu$-bounded intersection of $\{L(M_i)\}_{i\in\nat}$.
It is not difficult to show that the this intersection coincides with $Pal$.
\end{proof}

Lemma \ref{mu-bound} does not seem to be ``optimal'' in terms of $\mu$. What is the smallest possible choice of $\mu$ for $Pal$?
In a more general case, for various choices of $\mu$, what is the precise computational complexity of $\mu$-bounded intersections of dcf languages?

\let\oldbibliography\thebibliography
\renewcommand{\thebibliography}[1]{%
  \oldbibliography{#1}%
  \setlength{\itemsep}{-2pt}%
}
\bibliographystyle{alpha}

\begin{thebibliography}{99}
{\small

\bibitem{Ber79}
J. Berstel. Transductions and Context-Free Languages. Teubner Verlag, 1979.

\bibitem{Boa73}
L. Boasson, Un crit\`{e}re de rationalit\'{e} des langages alg\'{e}briques. In: Automata,
Languages and Programming, Nivat, M. (ed.), North-Holland Pubi. Comp.,
Amsterdam, 1973.

\bibitem{ER85}
A. Ehrenfeucht and G. Rozenberg. Strong iterative pairs and the regularity of context-free languages.
RAIRO--Informatique Th\'{e}orique 19 (1985) 43--56.

\bibitem{GHH76}
M. M. Geller, M. A. Harrison, and I. M. Havel. Normal forms of deterministic languages. Discrete Mathematics 16 (1976) 313--321.

\bibitem{GG66}
S. Ginsburg and S. Greibach. Deterministic context free languages. Information and Control 9 (1966) 620--648.

\bibitem{GS66}
S. Ginsburg and E. H. Spanier. Finite-turn pushdown automata. SIAM Journal on Computing 4 (1966) 429--453.

\bibitem{Har86}
M. A. Harrison. Iteration theorems for deterministic families of languages. Fundamenta Informaticae 9 (1986) 481--508.
A technical report version is available as Report No. UCB/CSD 86/271, Computer Science Division, University of California, Berkeley, California, 1985.

\bibitem{HH74}
M. A. Harrison and I. M. Havel. On the parsing of determinsitic languages. Journal of the ACM 21 (1974) 525--548.

\bibitem{Hib67}
T. N. Hibbard. A generalization of context-free determinism.
{Information and Control} 11 (1967) 196--238.

\bibitem{HU79}
J. E. Hopcroft and J. D. Ullman. Introduction to Automata Theory, Languages and Computation, Addison-Wesley Publishing Company, 1979.

\bibitem{Iga85}
Y. Igarashi. A pumping lemma for real-time deterministic context-free languages. Theoretical Computer Science 36 (1985) 89--97.

\bibitem{Kin80}
K. N. King. Iteration theorems for families of strict deterministic languages. Theoretical Computer Science 10 (1980) 317--333.

\bibitem{KMW08}
M. Kutrib, A. Malcher, and D. Wotschke. The Boolean closure of linear context-free languages. Acta Informatica 45 (2008) 177--191.

\bibitem{LW73}
L. Y. Liu and P. Weiner. An infinite hierarchy of intersections of context-free languages. Mathematical Systems Theory 7 (1973) 185--192

\bibitem{Lom70}
D. B. Lomet. A formalization of transition diagram systems. Journal of the  ACM 20 (1970) 235--237.

\bibitem{PP14}
G. Pighizzini and A. Pisoni. Limited automata and regular languages.
International Journal of Foundations of Computer Science 25 (2014) 897--916.

\bibitem{PP15}
G. Pighizzini and A. Pisoni. Limited automata and context-free languages. Fundamenta Informaticae 136 (2015) 157--176.

\bibitem{Rub18}
A. A. Rubtsov. A Structural lemma for deterministic context-free languages. In the Proc. of
the 22nd International Conference on Developments in Language Theory  (DLT 2018), Lecture Notes in Computer Science, vol. 11088,  pp. 553--565, Springer, 2018.

\bibitem{Sen90}
G. Senizergues. A Characterisation of deterministic context-free languages by means of right-congruences. Theoretical Computer Science 70 (1990) 213--232.

\bibitem{Wis76}
D. S. Wise. A strong pumping lemma for context-free languages. Theoretical Computer Science 3 (1976) 359--369.

\bibitem{Wot73}
D. Wotschke. The Boolean closure of the deterministic and nondeterministic context-free languages. In the Proc. of the 3rd Annual Symposium of ``Gesellschaft f\"{u}r Informatik'', Lecture Notes in Computer Science, vol. 1, pp.113--121, Springer, 1973.

\bibitem{Wot78}
D. Wotschke. Nondeterminism and Boolean operations in pda's. Journal of  Computer and System Sciences 16 (1978) 456--461.

\bibitem{Yam08}
T. Yamakami. Swapping lemmas for regular and context-free languages. Manuscript, available at  arXiv:0808.4122, 2008.

\bibitem{Yam14a}
T. Yamakami. Oracle pushdown automata, nondeterministic reducibilities, and the  hierarchy over the family of context-free languages.
In the Proc. of the 40th International Conference on Current Trends in Theory and Practice of Computer Science (SOFSEM 2014), Lecture Notes in Computer Science, vol.8327, pp.514--525, Springer, 2014. A complete version appeared at arXiv:1303.1717 under a slightly different title.

\bibitem{Yam16}
T. Yamakami. Pseudorandom generators against advised context-free languages. Theoretetical Computer Science 613 (2016) 1--27.

\bibitem{Yam19}
T. Yamakami. Behavioral strengths and weaknesses of various models of limited automata. In the Proc. of the 45th International Conference on Current Trends in Theory and Practice of Computer Science (SOFSEM 2019), Lecture Notes in Computer Science, vol. 11376, pp. 519--530, Springer, 2019. A complete and corrected version is available at arXiv:2111.05000, 2021.

\bibitem{Yam20}
T. Yamakami. Intersection and union hierarchies of deterministic context-free languages and pumping lemmas. In the Proceedings of the 14th International Conference on Language and Automata Theory and Applications (LATA 2020), Lecture Notes in Computer Science, vol. 12038, pp. 341--353, Springer, 2020.

\bibitem{Yam21}
T. Yamakami. The no endmarker theorem for one-way probabilistic pushdown automata. Manuscript, available at arXiv:2111.02688, 2021.

\bibitem{Yam21b}
T. Yamakami. Between SC and LOGDCFL: families of languages accepted by logarithmic-space deterministic auxiliary depth-k storage automata.
International Journal of Computer Mathematics: Computer Systems Theory 8, 1--31, 2023. A preliminary version apperead in the Proc. of COCOON 2021, LNCS, vol. 13025, pp. 164--175, Springer,  2021.

\bibitem{Yam22}
T. Yamakami. Nondeterministic auxiliary depth-bounded storage automata and semi-unbounded fan-in cascading circuits (extended abstract). In the Proc. of the 28th International Conference on Computing and Combinatorics (COCOON 2022), Lecture Notes in Computer Science, vol. 13595, pp. 61--69, Springer, 2022. A corrected and complete version is available at arXiv:2412.09186.

\bibitem{Yam25}
T. Yamakami. What is the most natural generalized pumping lemma beyond regular and context-free languages? In the Proc. of the 26th IFIP WG 1.02 International Conference on Descriptional Complexity of Formal Systems (DCFS 2025), Lecture Notes in Computer Science, vol. 15759, pp. 196--210, Springer, 2025.

\bibitem{Yu89}
S. Yu. A pumping lemma for deterministic context-free languages. Information Processing Letters 31 (1989) 47--51.

}
\end{thebibliography}

\end{document}